\pdfoutput=1 \pdfminorversion=5 \pdfobjcompresslevel=3 \pdfcompresslevel=9

\documentclass[11pt,letter]{article}

\usepackage[utf8]{inputenc}
\usepackage[pdftex]{graphicx}
\usepackage[usenames,dvipsnames,table]{xcolor}

\usepackage{lmodern}  \usepackage{microtype}\usepackage[in]{fullpage}

\usepackage{enumitem}

\usepackage{amssymb}
\usepackage{amsmath}
\usepackage{amsthm}\usepackage{float}
\usepackage{morefloats}
\usepackage{cite} 
\usepackage{thmtools}
\usepackage{thm-restate}

\usepackage{multirow}

\usepackage{adjustbox}
\usepackage{tikz}
\usetikzlibrary{positioning,calc}

\usepackage{todonotes}

\usepackage{environ} 
\makeatletter
\newsavebox{\box@tikzpicture}
\NewEnviron{tikzpicture*}[2][]{%
  \begin{lrbox}{\box@tikzpicture}%
    \begin{tikzpicture}[#1]
      \BODY
    \end{tikzpicture}%
  \end{lrbox}
  \pgfmathsetmacro\width@scale@picture{#2/\wd\box@tikzpicture}%
  \begin{tikzpicture}[#1,scale=\width@scale@picture]
    \BODY
  \end{tikzpicture}%
}%
\makeatother

\usepackage[noend]{algpseudocode} 
\usepackage{hyperref}

\usepackage{url}
\urldef{\maildiku}\path|roden@diku.dk|
\urldef{\mailjacob}\path|jaho@di.ku.dk|

\usepackage{accents}
\DeclareMathSymbol{\widehatsym}{\mathord}{largesymbols}{"62}
\makeatletter
\newcommand\lowerwidehatsym[1][-1.3ex]{%
  \text{\smash{\raisebox{#1}{$\widehatsym$}}}%
}
\newcommand\fixwidehat[1]{%
  \mathchoice%
  {\accentset{\displaystyle\lowerwidehatsym}{#1}}%
  {\accentset{\textstyle\lowerwidehatsym[-1.2ex]}{#1}}%
  {\accentset{\scriptstyle\lowerwidehatsym[-1.6ex]}{#1}}%
  {\accentset{\scriptscriptstyle\lowerwidehatsym}{#1}}%
}
\makeatother

\DeclareMathOperator*{\argmin}{arg\,min}

\newcommand{\Oo}{\mathcal O}  

\newcommand{\ceil}[1]{\left\lceil#1\right\rceil}
\newcommand{\floor}[1]{\left\lfloor#1\right\rfloor}
\newcommand{\set}[1]{\left\{#1\right\}}
\newcommand{\cond}{\mathrel{}\middle\vert\mathrel{}}
\newcommand{\sizeof}[1]{\left\lvert#1\right\rvert}
\newcommand{\sizefrac}[2]{\frac{\sizeof{#1}}{\sizeof{#2}}}

\newcommand{\dist}{\operatorname{dist}}

\newcommand{\norm}[1]{\lVert#1\rVert}

\newcommand{\alphamax}{\fixwidehat{\alpha}}
\newcommand{\Kmax}{\fixwidehat{K}}
\newcommand{\Smax}{\fixwidehat{S}}
\newcommand{\Cmax}{\fixwidehat{C}}

\newcommand{\Estar}{E^{*}}

\newcommand{\sstar}{S^{*}}
\newcommand{\sminus}{S^{-}}
\newcommand{\splus}{S^{+}}
\newcommand{\schange}{S^{\Delta}}
\newcommand{\cstar}{C^{*}}
\newcommand{\cminus}{C^{-}}
\newcommand{\cplus}{C^{+}}
\newcommand{\cchange}{C^{\Delta}}

\newcommand{\alphastar}{\alpha^{*}}
\newcommand{\alphaold}{\alpha^{\textsc{old}}}
\newcommand{\alphanew}{\alpha^{\textsc{new}}}

\newcommand{\load}{\ell}
\newcommand{\maxload}{L}
\newcommand{\Aa}{\mathcal A}
\newcommand{\cold}{C^{\textsc{old}}}
\newcommand{\cnew}{C^{\textsc{new}}}
\newcommand{\eold}{E^{\textsc{old}}}
\newcommand{\enew}{E^{\textsc{new}}}
\newcommand{\Gold}{G^{\textsc{old}}}
\newcommand{\Gnew}{G^{\textsc{new}}}
\newcommand{\Aaold}{\Aa^{\textsc{old}}}
\newcommand{\Aanew}{\Aa^{\textsc{new}}}

\newcommand{\opt}{\textsc{opt}}
\newcommand{\mystar}{\prime} 

\newcommand{\cm}{C_M}
\newcommand{\gm}{G_M}
\newcommand{\alpham}{\alpha_M}

\declaretheorem[style=plain,name={Theorem}]{theorem}
\declaretheorem[style=plain,name={Lemma},sibling=theorem]{lemma}
\declaretheorem[style=plain,name={Corollary},sibling=theorem]{corollary}
\declaretheorem[style=plain,name={Observation},sibling=theorem]{observation}
\declaretheorem[style=definition,name={Definition},sibling=theorem]{definition}
\declaretheorem[style=definition,name={Remark},sibling=theorem]{remark}

\usepackage{authblk}

\author[1]{Aaron Bernstein}
\author[2]{Jacob Holm\thanks{This research is supported by Mikkel Thorup's Advanced Grant DFF-0602-02499B from the Danish Council for Independent Research under the Sapere Aude research career programme.}}
\author[3]{Eva Rotenberg\thanks{This research was partly conducted during the third author's time as a PhD student at University of Copenhagen.}}
\affil[1]{Technical University of Berlin\\
  bernstei@gmail.com}
\affil[2]{University of Copenhagen (DIKU)\\
  jaho@di.ku.dk}
\affil[3]{Technical University of Denmark\\ 
	erot@dtu.dk}

\title{Online Bipartite Matching with Amortized $\Oo(\log^2 n)$ Replacements}

\begin{document}

\thispagestyle{empty}
\maketitle
\begin{abstract}
In the online bipartite matching problem with replacements, all the vertices on one side of the bipartition are given, and the vertices on the other side arrive one by one with all their incident edges.
The goal is to maintain a maximum matching while minimizing the number of changes (replacements)
to the matching. We show that the greedy algorithm that always takes the shortest augmenting path from the newly inserted vertex (denoted the SAP protocol) uses at most amortized $\Oo(\log^2 n)$ replacements per insertion, where $n$ is the total number of vertices inserted. This is the first analysis to achieve a polylogarithmic number of replacements for \emph{any} replacement strategy, almost matching the $\Omega(\log n)$ lower bound. The previous best strategy known achieved amortized $\Oo(\sqrt{n})$ replacements    
[Bosek, Leniowski, Sankowski, Zych, FOCS 2014]. For the SAP protocol in particular, nothing better than then trivial $\Oo(n)$ bound was known except in special cases.
 Our analysis immediately implies the same upper bound of $\Oo(\log^2 n)$ reassignments for the capacitated assignment problem, where each vertex on the static side of the bipartition is initialized with the capacity to serve a number of vertices.
  
We also analyze the problem of minimizing the maximum server load. We show that if the final graph
has maximum server load $L$, then the SAP protocol makes amortized $\Oo( \min\{L \log^2 n , \sqrt{n}\log n\})$ reassignments. We also show that this is close to tight because $\Omega(\min\{L, \sqrt{n}\})$ reassignments can be necessary.
\end{abstract}
\setcounter{page}{0}
\thispagestyle{empty}
\newpage

\section{Introduction}
In the online bipartite matching problem, the vertices on one side are given in advance
(we call these the servers $S$), while the vertices on the other side (the clients $C$) arrive
one at a time with all their incident edges. In the standard online model the arriving
client can only be matched immediately upon arrival, and the matching cannot be changed later.
Because of this irreversibility, the final matching might not be maximum; no algorithm can guarantee better than a $(1-1/e)$-approximation~\cite{Karp90}. But in many settings the irreversibility assumption
is too strict: rematching a client is expensive but not impossible. This motivates
the online bipartite matching problem with replacements, where the goal is to at all times match as many clients as possible, while minimizing the number of changes to the matching.
Applications include hashing, job scheduling, web hosting, streaming content delivery, and data storage; see~\cite{conf/infocom/ChaudhuriDKL09} for more details. 

In several of the applications above, a server can serve multiple clients, which raises the question of online bipartite \emph{assignment} with reassignments. There are two ways of modeling this:
\begin{description}
    \item[Capacitated assignments.] Each server $s$ comes with the capacity to serve some number of clients $u(s)$, where each $u(s)$ is given in advance. Clients should be assigned to a server, and at no times should the capacity of a server be exceeded. There exists an easy reduction showing that this problem is equivalent to online matching with replacements~\cite{BernsteinKPPS17}.
A more formal description is given in Section~\ref{sec:capacitated-assignment}.
    \item[Minimize max load.] There is no limit on the number of clients a server can serve, but we want to (at all times) distribute the clients as ``fairly'' as possible, while still serving all the clients. Defining the load on a server as the number of clients assigned to it, the task is to, at all times, minimize the maximum server load --- with as few reassignments as possible. A more formal
description is given in Section~\ref{sec:MinMax}
\end{description} 

While the primary goal is to minimize the number of replacements, special emphasis has been placed on
analyzing the \emph{SAP} protocol in particular, which always augments down a shortest augmenting path from the newly arrived client to a free server (breaking ties arbitrarily). This is the most natural
replacement strategy, and shortest augmenting paths are already of great interest in graph
algorithms: they occur in for example in Dinitz' and Edmonds and Karp's algorithm for maximum flow~\cite{Dinic70,Edmonds72}, and in Hopcroft and Karp's algorithm for maximum matching in bipartite graphs~\cite{Hopcroft73}.

Throughout the rest of the paper, we let $n$ be the number of clients in the final graph,
and we consider the \emph{total} number of replacements during the entire sequence of 
insertions; this is exactly $n$ times the amortized number of replacements. 
The reason for studying the vertex-arrival model (where each client arrives with all its incident edges) instead of the (perhaps more natural) edge-arrival model is the existence of a trivial lower bound of $\Omega(n^2)$ total replacements in this model: Start with a single edge, and maintaining at all times that the current graph is a path, add edges to alternating sides of the path. Every pair of
insertions cause the entire path to be augmented, leading to a total of $\sum_{i= 1}^{n/2} i \in \Omega(n^2)$ replacements.

\subsection{Previous work}
The problem of online bipartite matchings with replacements was introduced in 1995 by Grove, Kao, Krishnan, and Vitter~\cite{Grove:95}, who showed matching upper and lower bounds of $\Theta(n\log n)$ replacements for the case where each client has degree two. 
In 2009, Chadhuri, Daskalakis, Kleinberg, and Lin~\cite{conf/infocom/ChaudhuriDKL09}
showed that for any arbitrary underlying bipartite graph, if the client vertices arrive in a random order, the expected number of replacements (in their terminology, the \emph{switching cost}) is $\Theta(n\log n)$ using SAP, which they also show is tight. They also show that if the bipartite graph remains a forest, there exists an algorithm (not SAP) with $\Oo(n\log n)$ replacements, and a matching lower bound. 
Bosek, Leniowski, Sankowski and Zych later analyzed the SAP protocol 
for forests, giving an upper bound of $\Oo(n \log^2 n)$ replacements~\cite{Bosek2015}, later improved to the optimal $\Oo(n \log n)$ total replacements~\cite{BosekLZS17}. 
For general bipartite graphs, no analysis of SAP is known that shows better than the trivial $\Oo(n^2)$
total replacements. Bosek et al.~\cite{BosekLSZ14} showed a different algorithm that achieves a total of 
$\Oo(n \sqrt{n})$ replacements. They also show how to implement this algorithm in total
time $\Oo(m\sqrt{n})$, which matches the best performing combinatorial algorithm
for computing a maximum matching in a static bipartite graph (Hopcroft and Karp~\cite{Hopcroft73}).

The lower bound of $\Omega (\log n)$ by Grove et al.~\cite{Grove:95} has not been improved since, and is conjectured by Chadhuri et al.~\cite{conf/infocom/ChaudhuriDKL09} to be tight, even for SAP, in the general case. We take a giant leap towards closing that conjecture.

For the problem of minimizing maximum load,~\cite{Gupta:2014} and~\cite{BernsteinKPPS17} showed an approximation solution: with
only $\Oo(1)$ amortized changes per client insertion they maintain
an assignment $\Aa$ such that at all times the maximum load is within a factor of $8$ of optimum.

The model of online algorithms with replacements -- alternatively referred to as online algorithms with recourse -- has also been studied for a variety of problems other than matching. 
The model is similar to that of online algorithms, except that instead of trying 
to maintain the best possible approximation without making any changes, the goal is to maintain
an optimal solution while making as few changes to the solution as possible. This model encapsulates settings in which changes to the solution are possible but expensive. 
The model originally goes back to online Steiner trees~\cite{Imase:91}, and there
have been several recent improvements for online Steiner tree with recourse~\cite{Megow:16,Gu:2013:PDM:2488608.2488674,Gupta:2014:OST:2634074.2634108,Lacki:15}. There are many papers on online scheduling that try to minimize the number of job reassignments~\cite{Phillips1998,WESTBROOK:2000,Andrews1999,Sanders:09,Skutella:10,Epstein2011}. The model has also been studied in the context of flows~\cite{WESTBROOK:2000,Gupta:2014}, and there is a very recent result on online set cover with recourse~\cite{Gupta:2017}.

\subsection{Our results}

\begin{theorem}
    \label{thm:main-uncapacitated}
    SAP makes at most $\Oo(n\log^2 n)$ total replacements when $n$ clients are added.
\end{theorem}
This is a huge improvement of the $\Oo(n\sqrt{n})$ bound by~\cite{BosekLSZ14}, and is only a log factor from the lower bound of $\Omega(n\log n)$ by~\cite{Grove:95}. It is also a huge improvement of the analysis of SAP; previously no better upper bound than $\Oo(n^2)$ replacements for SAP was known.
To attain the result we develop a new tool for analyzing matching-related properties of graphs
(the balanced flow in Sections~\ref{sec:server-flow} and~\ref{sec:maxmatch}) that is quite general, 
and that we believe may be of independent interest.

Although SAP is an obvious way of serving the clients as they come, it does not immediately allow for an efficient implementation. Finding an augmenting path may take up to $\Oo(m)$ time, where $m$ denotes the total number of edges in the final graph. Thus, the naive implementation takes $\Oo(mn)$ total time. However, short augmenting paths can be found substantially faster, and 
using the new analytical tools developed in this paper, we are able to exploit this in a data structure that finds the augmenting paths efficiently:

\begin{restatable}{theorem}{implement}
    \label{thm:implement}
    There is an implementation of the SAP protocol that runs in total time $\Oo(m\sqrt{n}\sqrt{\log n})$. 
\end{restatable}
Note that this is only an $\Oo(\sqrt{\log n})$ factor from the offline algorithm of Hopcroft and Karp~\cite{Hopcroft73}. 
This offline algorithm had previously been matched in the online setting by the algorithm of Bosek et al.~\cite{BosekLSZ14}, 
which has total running time $O(m\sqrt{n})$. 
Our result has the advantage of combining multiple desired properties in a single algorithm: 
few replacements ($O(n\log^2(n))$ vs. $O(n^{1.5})$ in \cite{BosekLSZ14}), fast implementation ($\Oo(m\sqrt{n}\sqrt{\log n})$
vs. $\Oo(m\sqrt{n})$ in \cite{BosekLSZ14}), and the most natural augmentation protocol (shortest augmenting path).

Extending our result to the case where each server can have multiple clients, we use that the capacitated assignment problem is equivalent to that of matching 
(see Section~\ref{sec:capacitated-assignment} to obtain:

\begin{restatable}{theorem}{capacitate}
    \label{thm:main-capacitated}
    SAP uses at most $\Oo(n\log^2 n)$ reassignments for the capacitated assignment problem, where $n$ is the number of clients.
\end{restatable}

In the case where we wish to minimize the maximum load, such a small number of total reassignments is not possible.  Let $\opt(G)$ denote the minimum possible maximum load in graph $G$.  We present a lower bound showing
that when $\opt(G)=\maxload$ we may need as many as $\Omega(n \maxload)$ reassignments, as well as a nearly matching upper bound.

\begin{restatable}{theorem}{lowerbound}
    \label{thm:assignment-lower-bound}
    For any positive integers $n$ and $\maxload \leq \sqrt{n/2}$ divisible by $4$
    there exists a graph $G=(C\cup S,E)$ with $\sizeof{C} = n$ and $\opt(G) = \maxload$,
    along with an ordering in which the clients in $C$ are inserted, 
    such that any algorithm for the exact online assignment problem
    requires a total of $\Omega(n\maxload)$ changes.
    This lower bound holds even if the algorithm knows the entire graph $G$ in advance,
    as well as the order in which the clients are inserted.
\end{restatable}

\begin{restatable}{theorem}{minmax}
    \label{thm:assignment-upper-bound}
Let $C$ be the set of all clients inserted, let $n = \sizeof{C}$, and let $\maxload=\opt(G)$
be the minimum possible maximum load in the final graph $G = (C \cup S, E)$.
SAP at all times maintains an optimal assignment while making
a total of $\Oo(n\min \set{\maxload \,\smash{\log^2 n}, \sqrt{n}\log n})$ 
reassignments.
\end{restatable}

\subsection{High level overview of techniques}
\label{sec:overview}
Consider the standard setting in which we are given the entire graph from the beginning
and want to compute a maximum matching.
The classic shortest-augmenting paths algorithm constructs a matching by at
every step picking a shortest augmenting path in the graph. We now show a very simple
argument that the total length of all these augmenting paths is $\Oo(n\log n)$.
Recall the well-known fact 
that if all augmenting paths in the matching have length $\geq h$, then
the current matching is at most $2n/h$ edges from optimal~\cite{Hopcroft73}.
Thus the algorithm augments down at most $2n/h$ augmenting paths of length $\geq h$. 
Let $P_1, P_2, ..., P_k$
denote all the paths augmented down by the algorithm in decreasing order of $\sizeof{P_i}$;
then $k \leq n$, and $\sizeof{P_i}=h$ implies $i\leq2n/h$. But then $\sizeof{P_i} \leq 2n/i$,
so $\sum_{1 \leq i \leq k} \sizeof{P_i} \leq 2n\sum_{1\leq i\leq k}\frac{1}{i} = 2n(\ln(k)+\Oo(1)) = \Oo(n\log k) = \Oo(n\log n)$.

In the online setting, the algorithm does not have access to the entire graph.
It can only choose the shortest augmenting path from the arriving client $c$.
We are nonetheless able to show a similar bound for this setting:

\begin{restatable}{lemma}{longpaths}
\label{lem:bound-long-paths}
Consider the following protocol for constructing a matching: For each client $c$ in arbitrary order, augment along the shortest
augmenting path from $c$ (if one exists).
Given any $h$, this protocol augments down a total of at most 
$4n\ln(n)/h$ augmenting paths of length $> h$. 
\end{restatable}

The proof of our main theorem then follows directly from the lemma.

\begin{proof}[{\bf Proof of Theorem~\ref{thm:main-uncapacitated}}]
  Note that the SAP protocol exactly follows the condition of Lemma~\ref{lem:bound-long-paths}. 
  Now, Given any $0 \leq i \leq \log_2(n) + 1$, we say that an augmenting path 
  is at level $i$ if its length is in the interval $[2^i, 2^{i+1})$.
    By Lemma~\ref{lem:bound-long-paths}, the SAP protocol augments down at most
    $4n\ln(n)/2^i$ paths of level $i$. Since each of those paths contains at most
    $2^{i+1}$ edges, the total length of augmenting paths of level $i$
    is at most $8n\ln(n)$. Summing over all levels yields the desired $\Oo(n\log^2 n)$ bound.
\end{proof}

The entirety of Sections~\ref{sec:server-flow} and~\ref{sec:maxmatch} is devoted to proving Lemma
\ref{lem:bound-long-paths}. Previous algorithms attempted to bound
the total number of reassignments
by tracking how some property of the matching $M$ changes over time.
For example, the analysis of Gupta et al.~\cite{Gupta:2014} keeps track of
changes to the "height" of vertices in $M$,
while the algorithm with $\Oo(n\sqrt{n})$ reassignments~\cite{BosekLSZ14} 
takes a more direct approach, and 
uses a non-SAP protocol whose changes to $M$ depend on 
how often each particular client has already been reassigned. 

Unfortunately such arguments have had limited success because the matching $M$
can change quite erratically. This is especially true under the SAP protocol, which
is why it has only been analyzed in very restrictive settings~\cite{conf/infocom/ChaudhuriDKL09,Grove:95,Bosek2015}.
We overcome this difficulty by showing that it is enough to analyze how new clients
change the structure of the graph $G = (C \cup S, E)$, without reference to any particular matching.

Intuitively, our analysis keeps track of  how "necessary" 
each server $s$ is (denoted $\alpha(s)$ below). So for example,
if there is a complete bipartite graph with 10 servers and 10 clients, then all servers
are completely necessary. But if the complete graph has 20 servers and 10 clients, then while any matching has 10 matched servers and 10 unmatched ones, 
it is clear that if we abstract away from the particular matching every server is 1/2-necessary.
Of course in more complicated graphs different servers might have different necessities, and
some necessities might be very close to 1 (say $1 - 1/n^{2/3}$).
Note that server necessities depend only on the graph, not on any particular matching.
Note also that our algorithm never
computes the server necessities, as they are merely an analytical tool.

We relate necessities to the number of reassignments with 2 crucial arguments. 
{\bf 1.} Server necessities only increase as clients are inserted, and once a server has 
$\alpha(s) = 1$, then regardless of the current matching, no future augmenting path will go through $s$.
{\bf 2.} If, \emph{in any matching}, the shortest augmenting path from a new client $c$ is long,
then the insertion of $c$ will increase the necessity of servers that already had high necessity.
We then argue that this cannot happen too many times before the servers involved have necessity 1,
and thus do not partake in any future augmenting paths.

\subsection{Paper outline}
In Section~\ref{sec:prelim}, we introduce the terminology necessary to understand the paper. In Section~\ref{sec:server-flow}, we introduce and reason about the abstraction of a balanced server flow, a number that reflects the necessity of each server. In Section~\ref{sec:maxmatch}, we use the balanced server flow to prove Lemma~\ref{lem:bound-long-paths}, which proves our main theorem that SAP makes a total of $\Oo(n \log^2 n)$ replacements. In Section~\ref{sec:implementation}, we give an efficient implementation of SAP. Finally, in Section~\ref{sec:extend}, we present our results on capacitated online assignment, and for minimizing maximum server load in the online assignment problem.

\section{Preliminaries and notation}\label{sec:prelim}

Let $(C,S)$ be the vertices, and $E$ be the edges of a bipartite graph.  We call $C$ the \emph{clients}, and $S$ the \emph{servers}.  Clients arrive, one at a time, and we must maintain an explicit maximum matching of the clients. For simplicity of notation, we assume for the rest of the
paper that $C \neq \emptyset$.
For any vertex $v$, let $N(v)$ denote the neighborhood of $v$, and for any $V\subseteq C\cup S$ let $N(V)=\bigcup_{v\in V}N(v)$.

\begin{theorem}[Halls Marriage Theorem~\cite{JLMS:JLMS0026}]\label{thm:hall}
  There is a matching of size $\sizeof{C}$   if and only if
  $\sizeof{K}\leq\sizeof{N(K)}$ for all $K\subseteq C$.
\end{theorem}

\begin{definition}
    \label{dfn:augmenting-path}
    Given any matching in a graph $G = (C \cup S, E)$, an alternating path is one which alternates between unmatched and matched edges.     An augmenting path is an alternating path that starts and ends with an unmatched vertex. Given any augmenting path $P$, ``flipping'' the matched status of every edge on $P$ gives a new larger matching. We call this process \emph{augmenting down $P$}.
\end{definition}

Denote by SAP the algorithm
that upon the arrival of a new client $c$ augments down the shortest augmenting path from $c$;
ties can be broken arbitrarily, and if no augmenting path from $c$ exists the algorithm does nothing.
Chaudhuri et al.~\cite{conf/infocom/ChaudhuriDKL09} 
showed that if the final graph contains a perfect matching, then
the SAP protocol also returns a perfect matching. We now generalize this as follows

\begin{observation}
\label{obs:ignore-clients}
Because of the nature of augmenting paths, once a client $c$ or a server $s$ is matched by the SAP protocol, it will remain matched
during all future client insertions. On the other hand, if a client $c$ arrives and there is
no augmenting path from $c$ to a free server, then during the entire sequence of client insertions $c$ will never be matched by the SAP protocol; no alternating path can go through $c$ because it
is not incident to any matched edges.
\end{observation}

\begin{lemma}
\label{lem:sap-correctness}
The SAP protocol always maintains a maximum matching in the current graph $G = (C \cup S, E)$.
\end{lemma}

\begin{proof}
Consider for contradiction the first client $c$ such that after the insertion of $c$, 
the matching $M$ maintained
by the SAP protocol is not a maximum matching. Let $C$ be the set of clients before $c$ was inserted.
Since $M$ is maximum in the graph $G = (C \cup S,E)$ but not in 
$G' = (C \cup S \cup \set{c},E)$,
it is clear that $c$ is matched in the maximum matching $M'$ of $G'$ but not in $M$. 
But this contradicts the well known property of augmenting paths that the symmetric difference 
$M \oplus M'$ contains an augmenting path in $M$ from $c$ to a free server.
\end{proof}

\section{The server flow abstraction}
\label{sec:server-flow}

\subsection{Defining the Server Flow}
We now formalize the notion of server necessities from Section~\ref{sec:overview}
by using a flow-based notation. 
The necessity of a server $s$ will be the value $\alpha (s)\in \left[0,1\right]$ of a balanced server flow $\alpha$: We will now go on to define a server flow, define what it means for a server flow to be balanced, and then, show that the balanced server flow is unique.
\begin{definition}\label{def:flow}
 Given any graph $G = (C \cup S, E)$, define a \emph{server flow} $\alpha$ as any map from $S$ to the 
nonnegative reals such that there exist nonnegative $(x_e)_{e\in E}$ with:
  \begin{align*}
    \forall c\in C&: \sum_{s\in N(c)} x_{cs} = 1
    &
    \forall s\in S&: \sum_{c\in N(s)} x_{cs} = \alpha(s)
  \end{align*}
  We say that such a set of $x$-values \emph{realize} the server flow.
\end{definition}

A server flow can be thought of as a fractional assignment from $C$ to $S$; note, however,
that is is not necessarily a fractional matching, since servers may have a load greater than $1$.
Note also that the same server flow may be realized in more than one way.
Furthermore, if $\sizeof{N(c)}=0$ for some $c\in C$ then $\sum_{s\in N(c)}x_{cs}=0\neq 1$, so no server flow is possible.  So suppose (unless otherwise noted) that $\sizeof{N(c)}\geq 1$ for all $c\in C$.

The following theorem can be seen as a generalization of Hall's Marriage Theorem:
\begin{lemma}\label{lem:serverflowexists}
  If $\max_{\emptyset\subset K\subseteq C}\frac{\sizeof{K}}{\sizeof{N(K)}} = \frac{p}{q}$, then there exists a server flow where every server $s\in S$ has $\alpha(s) \leq \frac{p}{q}$.
\end{lemma}
\begin{proof}
  Let $\cstar$ be the original set $C$ but with $q$ copies of each client. Similarly let $\sstar$ contain $p$ copies of each server, and let $\Estar$ consist of all $pq$ edges between copies of the endpoints of each edge in $E$.

  Now let $K^*\subseteq \cstar$, and let $K\subseteq C$ be the originals that the vertices in $K^*$ are copies of.  Then $\sizeof{K^*}\leq q\sizeof{K}\leq p\sizeof{N(K)}=\sizeof{N(K^*)}$, so the graph $(\cstar\cup\sstar,\Estar)$ satisfies Hall's theorem and thus it has some matching $M$ in which every client in $\cstar$ is matched.
    Now, for $cs\in E$ let 
  \begin{align*}
    x_{cs}&=\tfrac{1}{q}\sizeof{\set{c^*s^*\in M\cond c^*\text{ is a copy of }c\text{ and }s^*\text{ is a copy of }s\strut\!}}
  \end{align*}
  Since for each $c\in C$ all $q$ copies of $c$ are matched, $\sum_{s\in N(c)}x_{cs}=\frac{q}{q}=1$ for all $c\in C$.  Similarly, since for each $s\in S$ at most $p$ copies of $s$ are matched, $\sum_{c\in N(s)}x_{cs}\leq \frac{p}{q}$.
Thus, $(x_e)_{e\in E}$ realizes the desired server flow.
\end{proof}

\begin{definition}\label{def:balanced}
  We say that a server flow $\alpha$ is \emph{balanced}, if additionally:
  \begin{align*}
    &\forall c\in C, s\in N(c)\setminus A(c): x_{cs} = 0
    &\text{where }A(c) &= \argmin_{s\in N(c)} \alpha(s)
  \end{align*}
  That is, if each client only sends flow to its least loaded neighbours.

  We call the set $A(c)$ the \emph{active} neighbors of $c$, and we call an edge $cs$ \emph{active} when $s \in A(c)$.
  We extend the definition to sets of clients in the natural way, so for $K\subseteq C$, $A(K)=\bigcup_{c\in K}A(c)$.
\end{definition}

\subsection{Uniqueness of Loads in a Balanced Server Flow}
\label{subsec:uniqueness}
Note that while there may be more than one server flow, we will show that the balanced server flow 
$\alpha$ is unique, although there may be many possible $x$-values $x_{cs}$ that realize $\alpha$.

\begin{lemma}\label{lem:alphaunique}
	A unique balanced server flow exists if and only if $\sizeof{N(c)}\geq1$ for all $c \in C$.
\end{lemma}

Clearly, it is necessary for all clients to have at least one neighbor for a server flow to exist, so the ``only if" part is obvious.
We dedicate the rest of this section to proving that this condition is sufficient.
In fact, we provide two different proofs of uniqueness, the first of which
is simpler but provides less intuition for what the unique $\alpha(s)$
values signify about the structure of the graph. 

\subsubsection{Short proof of Lemma~\ref{lem:alphaunique} via convex optimization}
\label{subsubsec:unique}
It is not hard to prove uniqueness by showing that a balanced server flow corresponds to the solution
to a convex program\footnote{The authors thank Seffi Naor for pointing this out to us.}.
Consider the convex optimization problem where the constraints
are those of a \emph{not necessarily balanced} server flow (Definition ~\ref{def:flow}),
and the objective function we seek to minimize is the sum of the squares of the server loads.

To be precise, the convex program contains a variable $\alpha_s$ for each server $s\in S$, and a variable $x_{cs}$ for each edge $\left(c,s\right)$ in the graph. Its objective is to minimize the function $\sum_{s\in S} \alpha_s^2$ subject to the constraints:
\begin{align*}
 0 & \leq x_{cs} \leq 1 &
\forall c\in C: & \sum_{s \in N(c)} x_{cs} = 1 &
\forall s\in S: & \sum_{c\in N(s)} x_{cs} = \alpha_s
\end{align*}

It is easy to check that because we introduce a separate variable $\alpha_s$ for each server load, the objective function is strictly convex, 
so the convex program has a unique minimum with respect to the server loads $\alpha_s$ (but not the edge flows).

We now observe that this unique solution is a \emph{balanced} server flow:
the constraints of the convex program ensure that it is a server flow, 
and were it not balanced, there would be some client $c$ that sends non-zero flow 
to both $s$ and $s^{\prime}$ where $\alpha(s)<\alpha(s^\prime)$, which would
be a contradiction because we can decrease the objective function by 
increasing $x_{cs}$ and decreasing $x_{cs^{\prime}}$. 
We have thus proved the existence of a balanced server flow.
 
We must now prove uniqueness, i.e. that all balanced server flows have the same server loads.
We will do this by showing that any balanced server flow optimizes the objective function of the convex function.
There are many standard approaches for proving this claim, but the simplest one we know of is based on existing literature 
on load balancing with selfish agents. 
In particular, we rely on the following simple auxiliary lemma, which is a simplified version of Lemma 2.2 in \cite{SuriTZ07}.

\begin{lemma}
\label{lem:dot-product}
Consider any balanced server flow $x_{cs}$, let $\alpha_s = \sum_{c \in C}x_{cs}$ be the server flow of $s$. Let $x^{\mystar}_{cs}$ be \emph{any} feasible
server flow, and let $\alpha_s^{\mystar} = \sum_{c \in C}x^{\mystar}_{cs}$ be the resulting server loads. Then, we always have:

$$\sum_{s \in S} \alpha^2_s \leq \sum_{s \in S} \alpha_s \alpha^{\mystar}_s $$

\end{lemma}

\begin{proof}
For any client $c$, let $\mu(c)$ ($\mu$ for minimum) be the minimum server load neighboring $c$ in the balanced solution $x_{cs}$. That is,
$\mu(c) = \min_{s \in N(c)} \alpha_s$. We then have

$$\sum_{s \in S} \alpha^2_s = \sum_{s \in S} \sum_{c \in C} x_{cs} \alpha_s = \sum_{c \in C} \sum_{s \in S} x_{cs} \alpha_s = 
\sum_{c \in C} \sum_{s \in S}  x_{cs} \mu(c) = \sum_{c \in C} \mu(c),$$

where the last inequality follows from the fact that each client sends one unit of flow, and the before-last inequality follows
from the fact that the flow is balanced, so for any edge $(c,s) \in E$ with $x_{cs} \neq 0$ we have $\alpha_s = \mu(c)$.

From the definition of $\mu(c)$ it follows that for \emph{any} edge $(c,s) \in E$, we have $\alpha_s \geq \mu(c)$. 
This yields:

$$\sum_{s \in S} \alpha^{\mystar}_s \alpha_s  = \sum_{s \in S} \sum_{c \in C} x^{\mystar}_{cs} \alpha_s = \sum_{c \in C} \sum_{s \in S} x^{\mystar}_{cs} \alpha_s
\geq \sum_{c \in C} \sum_{s \in S}  x^{\mystar}_{cs} \mu(c) = \sum_{c \in C} \mu(c).$$

We thus have $\sum_{s \in S} \alpha^2_s = \sum_{c \in C} \mu(c)$ and $\sum_{s \in S} \alpha^{\mystar}_s \alpha_s \geq \sum_{c \in C} \mu(c)$,
which yields the lemma.
\end{proof}

We now argue that any balanced flow is an optimal solution to the convex program, and is thus unique.
Consider any balanced flow with loads $\alpha_s$. 
To show that $\alpha_s$ is optimum, we need to show that for any feasible solution $\alpha^{\mystar}_s$
we have $\sum_{s \in S} \alpha^2_s \leq \sum_{s \in S}(\alpha^{\mystar})^2_s$.
Equivalently, let $\boldsymbol{\alpha}$ and $\boldsymbol{\alpha^{\mystar}}$ be the vectors of server loads in the two solutions.
We want to show that $\norm{\boldsymbol{\alpha}} \leq \norm{\boldsymbol{\alpha^{\mystar}}}$. 
This follows trivially from Lemma \ref{lem:dot-product}, which is equivalent to $\norm{\boldsymbol{\alpha}}^2 \leq \boldsymbol{\alpha} \cdot \boldsymbol{\alpha^{\mystar}}$.

\subsubsection{Longer combinatorial proof of uniqueness}
Although the reduction to convex programming is the most direct proof of uniqueness, it 
has the disadvantage of not providing any insight into what the unique $\alpha(s)$
values actually correspond to. We thus provide a more complicated combinatorial proof which
shows that the $\alpha(s)$ correspond to a certain hierarchical decomposition of the graph.

The following lemma will help us upper and lower bound the sum of flow to a subset of servers.

\begin{lemma}\label{lem:alphasum}
  If $\alpha$ is a balanced server flow, then
  \begin{align*}
    \forall T\subseteq S:
    \sizeof{\set{c\in C\cond A(c)\subseteq T}\strut\!}
    \leq
    \sum_{s\in T}\alpha(s)
    \leq
    \sizeof{\set{c\in C\cond A(c)\cap T\neq\emptyset}\strut\!}
  \end{align*}
\end{lemma}
\begin{proof}
  The first inequality is true because each client in the first set contributes exactly one to the sum (but there may be other contributions).  The second inequality is true because every client contributes exactly one to $\sum_{s\in S}\alpha(s)$, and the inequality counts every client that contributes anything to $\sum_{s\in T}\alpha(s)$ as contributing one.
\end{proof}

The first step to proving that every graph has a unique server flow $\alpha$ is to show
that the maximum value $\alphamax=\max_{s\in S}\alpha(s)$ is uniquely defined.
We start by showing that the generalization of Hall's Marriage Theorem from Lemma~\ref{lem:serverflowexists} is ``tight'' for a balanced server flow in the sense that there does indeed exist a set of $p$ clients 
with neighbourhood of size $q$ realizing the maximum $\alpha$-value $\frac{p}{q}$. In fact, the maximally necessary servers and their active neighbours (defined below) form such a pair of sets:

\begin{lemma}\label{lem:alphaset}
  Let $\alpha$ be a balanced server flow, let $\alphamax=\max_{s\in S}\alpha(s)$ be the maximal necessity, let $\Smax=\set{s\in S\cond \alpha(s)=\alphamax}$ be the maximally necessary servers, and let $\Kmax=\set{c\in C\cond A(c)\cap \Smax\neq\emptyset}$ be their active neighbours.  Then $N(\Kmax)=\Smax$ and $\sizeof{\Kmax}=\alphamax\sizeof{\Smax}$.
\end{lemma}
\begin{proof}
  Let $K=\set{c\in C\cond A(c)\subseteq \Smax}$, and note that $K \subseteq \Kmax$. However, we also have $\Kmax\subseteq K$: By definition of $\Smax$, and since we assume the server flow is balanced, $\Kmax\neq\emptyset$, and for every $c\in \Kmax$,  $N(c)=A(c)\subseteq\Smax$. 
Thus, $K = \Kmax$ and $N(\Kmax)=\Smax$. 
Now, note that by Lemma~\ref{lem:alphasum} 
  \begin{align*}
  \sizeof{\Kmax}=\sizeof{K}\leq\alphamax\sizeof{\Smax}\leq\sizeof{\Kmax}.
  \tag*{\qedhere}
  \end{align*} 
\end{proof}

We can thus show that $\alphamax$ exactly equals the maximal quotient $\sizefrac{K}{N(K)}$ over subsets $K$ of clients.

\begin{lemma}\label{lem:maxalpha}
  Let $\alpha$ be a balanced server flow, and let $\alphamax=\max_{s\in S}\alpha(s)$. Then 
  $$\alphamax=\max_{\emptyset\subset K\subseteq C}\frac{\sizeof{K}}{\sizeof{N(K)}}$$
  Furthermore, for any $K\subseteq C$, if $\sizeof{K}=\alphamax\sizeof{N(K)}$, then $\alpha(s)=\alphamax$ for all $s\in N(K)$.
\end{lemma}
\begin{proof}
  By definition of server flow, for $K\subseteq C$, $\sizeof{K}\leq\sum_{s\in N(K)}\alpha(s)\leq\alphamax\sizeof{N(K)}$, so $\sizeof{K}\leq\alphamax\sizeof{N(K)}$.
  Let $\Kmax$ be defined as in Lemma~\ref{lem:alphaset}.  Then $\alphamax=\frac{\sizeof{\Kmax}}{\sizeof{N(\Kmax)}\strut}\leq \max_{\emptyset\subset K\subseteq C}\frac{\sizeof{K}}{\sizeof{N(K)}}\leq\alphamax$.
  Finally, if $\sizeof{K}=\sum_{s\in N(K)}\alpha(s)=\alphamax\sizeof{N(K)}$ then $\alpha(s)\leq\alphamax$ for all $s\in S$ implies $\alpha(s)=\alphamax$ for $s\in N(K)$.
\end{proof}

\begin{corollary}\label{cor:alphasinglevalue}
  If $\max_{\emptyset\subset K\subseteq C}\frac{\sizeof{K}}{\sizeof{N(K)}}=\frac{\sizeof{C}}{\sizeof{S}}$ there is a unique balanced server flow.
\end{corollary}
\begin{proof}
  By Lemma~\ref{lem:serverflowexists} there exists a server flow with $\alpha(s)\leq\frac{\sizeof{C}}{\sizeof{S}}$ for all $s\in S$.  Since $\sum_{s\in S}\alpha(s)=\sizeof{C}$, any such flow must actually have $\alpha(s)=\frac{\sizeof{C}}{\sizeof{S}}$ for all $s\in S$, and be balanced.
  Uniqueness follows from Lemma~\ref{lem:maxalpha}.
\end{proof}

We are now ready to give a combinatorial proof of uniqueness. We will do so by showing that the $\alpha(s)$ in fact
express a very nice structural property of the graph, which can be thought of as a hierarchy of 
''tightness" for the Hall constraint.
As shown in Lemma \ref{lem:maxalpha},
the maximum $\alpha$ value $\alphamax$ corresponds to the tightest Hall constraint,
i.e. the maximum possible value of $\sizeof{K}/\sizeof{N(K)}$. Now, there may be many sets $K$
with $\sizeof{K}/\sizeof{N(K)} = \alphamax$, so let $\Cmax$ be a maximal such set; we will show that $\Cmax$ is in fact the union
of all sets $K$ with $\sizeof{K}/\sizeof{N(K)} = \alphamax$. 
Now, by Lemma \ref{lem:maxalpha}, every server 
$s \in N(\Cmax)$ has $\alpha(s) = \alphamax$. We will show that in fact, because $\Cmax$ captured
\emph{all} sets with tightness $\alphamax$, all servers $s \notin N(\Cmax)$ have $\alpha(s) < \alphamax$.
Thus, because the flow is balanced, all active edges incident to $\Cmax$ or $\Smax$ will be between
$\Cmax$ and $\Smax$; there will be no active edges coming from the outside. For this reason,
any balanced server flow $\alpha$ on $G = (C \cup S)$ can be the thought of as the union of two completely 
independent server flows: the first (unique) flow assigns 
$\alpha(s) = \alphamax = \sizeof{\Cmax} / \sizeof{N(\Cmax)}$ to all $s \in \Smax$, 
while the second is a balanced server flow on the remaining graph $G \setminus (\Cmax \cup \Smax)$.
Since this remaining graph is smaller, we can use induction on the size of the graph to claim that this second balanced server flow has unique $\alpha$-values, which completes the proof of uniqueness. 
If we follow through the entire inductive chain, we end up with a hierarchy of $\alpha$-values, which can be viewed as the result
of the following peeling procedure: first find the (maximally large) set $C_1$ that maximizes $\alpha_1 = \sizeof{C_1}/\sizeof{N(C_1}$ and assign every server $s \in N(C_1)$ a value of $\alpha_1$; then peel off $C_1$ and $N(C_1)$, find the (maximally large) set $C_2$ in the remaining graph that maximizes $\alpha_2 = \sizeof{C_2}/\sizeof{N(C_2}$, and assign every server $s \in N(C_2)$ value $\alpha_2$; peel off $C_2$ and $N(C_2)$ and continue in this fashion, until every server has some value $\alpha_i$. These values 
$\alpha_i$ assigned to each server are precisely the unique $\alpha(s)$ in a balanced server flow.

\begin{remark}
We were unaware of this when submitting the extended abstract, but a
similar hierarchical decomposition was used earlier to compute 
an approximate matching in the
semi-streaming setting: see \cite{GoelKK12}, \cite{Kapralov13}. Note that unlike
those papers, we do not end up relying on this decomposition for our main arguments.
We only present it here to give a combinatorial alternative to the convex 
optimization proof above: regardless of which proof we use, once uniqueness is 
established, the rest of our analysis is expressed purely in terms of balanced
server flows.
\end{remark}

\begin{proof}[Proof of Lemma~\ref{lem:alphaunique}]
  As already noted, $\sizeof{N(c)}\geq 1$ for all $c\in C$ is a necessary condition.
  We will prove that it is sufficient by induction on $i=\sizeof{S}$.  If $\sizeof{S}=1$, the flow $\alpha(s)=\sizeof{C}$ for $s\in S$ is trivially the unique balanced server flow.  Suppose now that $i>1$ and that it
  holds for all $\sizeof{S}<i$.  Now let $\alphamax=\max_{\emptyset\subset K\subseteq C}\frac{\sizeof{K}}{\sizeof{N(K)}}$ and let   \begin{align*}
                                            \Cmax &= \bigcup_{K\in\mathcal{K}}K &\text{where }\mathcal{K}=\set{K\subseteq C\cond \sizeof{K}=\alphamax\sizeof{N(K)}\strut
    }
  \end{align*}
  Note that for any $K_1,K_2\in\mathcal{K}$ we have 
  \begin{align*}
    \alphamax\sizeof{N(K_1\cup K_2)} &\geq \sizeof{K_1\cup K_2}\tag{by definition of $\alphamax$}\\
    &= \sizeof{K_1}+\sizeof{K_2}-\sizeof{K_1\cap K_2} \\
    &= \alphamax\sizeof{N(K_1)}+\alphamax\sizeof{N(K_2)}-\sizeof{K_1\cap K_2} \tag{since $K_1,K_2\in\mathcal{K}$}\\
    &\geq \alphamax\sizeof{N(K_1)}+\alphamax\sizeof{N(K_2)}-\alphamax\sizeof{N(K_1\cap K_2)} \tag{by definition of $\alphamax$}\\
    & \geq\alphamax\sizeof{N(K_1\cup K_2)} \tag{since $\sizeof{N(\cdot)}$ is submodular}
  \end{align*} so $K_1\cup K_2\in\mathcal{K}$ and thus $\Cmax\in\mathcal{K}$ and
    $\sizeof{\Cmax}=\alphamax\sizeof{N(\Cmax)}$. 
  If $N(\Cmax)=S$ then $\Cmax=C$ (otherwise 
$\frac{\sizeof{\Cmax}}{\sizeof{N(\Cmax)}\strut} < \frac{\sizeof{C}}{\sizeof{S}} \leq \alphamax$) and by Corollary~\ref{cor:alphasinglevalue} we are done, so suppose $\emptyset\subset N(\Cmax)\subset S$.  Consider the subgraph $G_1$ induced by $\Cmax\cup N(\Cmax)$ and the subgraph $G_2$ induced by $(C\setminus \Cmax)\cup (S\setminus N(\Cmax))$.

  By Corollary~\ref{cor:alphasinglevalue}, $G_1$ has a unique balanced server flow $\alpha_1$ with 
        $\alpha_1(s)=\alphamax$ for all $s\in N(\Cmax)$. 

 By our induction hypothesis, $G_2$ also has a \emph{unique} balanced server flow $\alpha_2$.

We proceed to show that the combination of $\alpha_1$ with $\alpha_2$ constitutes a unique balanced flow $\alpha$ of the entire graph $G$, defined as follows:
$$\alpha(s)=
  \begin{cases}
    \alpha_1(s)&\text{if }s\in N(\Cmax)\\
    \alpha_2(s)&\text{otherwise}
  \end{cases}$$
Note first that $\alpha$ is a balanced server flow for $G$, because both $G_1$ and $G_2$ have a set of $x$-values that realize them, and by construction these values (together with zeroes for each edge between $C\setminus \Cmax$ and $N(\Cmax)$) realize a balanced server flow for $G$.

For uniqueness, note that by Lemma~\ref{lem:maxalpha} any balanced server flow 
$\alpha'$ for $G$ must have $\alpha'(s)=\alphamax=\alpha_1(s)$ for $s\in N(\Cmax)$. We now show that for any $s \in S\setminus N(\Cmax)$, any balanced server flow $\alpha'$
must also have $\alpha'(s) = \alpha_2(s)$; then, the uniqueness of $\alpha$ will follow 
from the uniqueness of $\alpha_1$ and $\alpha_2$. 

Let $\Smax=\set{s\in S\cond \alpha'(s)=\alphamax}$ be the set of maximally necessary servers, and let 
$\Kmax=\{c\in C\mid $
$A(c)\cap\Smax\neq\emptyset\}$ be the set of clients with a maximally necessary server in their active neighborhood.  We will show that $\Kmax=\Cmax$. 
\begin{description}
    \setlength\itemsep{0em}
    \item[``$\subseteq$''] By Lemma~\ref{lem:alphaset}, $\sizeof{\Kmax}=\alphamax\sizeof{N(\Kmax)}$ so by definition of $\Cmax$, $\Kmax\subseteq\Cmax$.
    \item[``$\supseteq$''] On the other hand, $\sizeof{\Cmax}=\alphamax\sizeof{N(\Cmax)}$ so by Lemma~\ref{lem:maxalpha} we have $N(\Cmax)\subseteq\Smax$ and in particular $A(c)\subseteq\Smax$ for $c$ in $\Cmax$ and thus $\Cmax\subseteq\Kmax$.
\end{description}

Thus, by definition of $\Kmax$,  $A(c)\cap\Smax=\emptyset$ for all $c\in C\setminus\Cmax$.  And there are clearly no edges between $\Cmax$ and $S\setminus N(\Cmax)$.
But then, for any $(x_e)_{e\in E}$ realizing $\alpha'$,
the subset $(x_{cs})_{c\in C\setminus\Cmax,s\in S \setminus N(\Cmax)}$ realizes a balanced server flow in $G_2$, so since $\alpha_2$ is the unique balanced server flow in $G_2$ we have
$\alpha'(s)=\alpha_2(s)$ for $s\in S\setminus N(\Cmax)$.
\end{proof}

\subsection{How Server Loads Change as New Clients are Inserted}

From now on, let $\alpha$ denote the unique balanced server flow.
We want to understand how the balanced server flow changes as new clients are added.  For any server $s$, let $\alphaold(s)$ be the flow in $s$ \emph{before} the insertion of $c$,
and let $\alphanew(s)$ be the flow \emph{after}.  Also, let $\Delta\alpha(s)=\alphanew(s)-\alphaold(s)$.

Intuitively, as more clients are added
to the graph, the flow on the servers only increases, so no $\alpha(s)$ ever decreases. 
We now prove this formally.

\begin{lemma}\label{lem:nondecreasing}
  When a new client $c$ is added, $\Delta\alpha(s)\geq0$ for all $s\in S$.
\end{lemma}
\begin{proof}
  Let $\sstar=\set{s \in S \cond \alphanew(s) < \alphaold(s)}$. We want to show that $\sstar=\emptyset$.  Say for contradiction that $\sstar\neq\emptyset$, and let $\alphastar = \min_{s \in \sstar} \alphanew(s)$.
We will now partition $S$ into three sets.
\begin{align*}
  \sminus &= \set{s \in S \cond \alphaold(s) \leq \alphastar}
  \\
  \schange &= \set{s \in S \cond  \alphaold(s) > \alphastar \land \alphanew(s) = \alphastar}
  \\
  \splus &= \set{s \in S \cond \alphaold(s) > \alphastar \land \alphanew(s) > \alphastar}
\end{align*}
It is easy to see that these sets form a partition of $S$, and that  $\emptyset\neq\schange\subseteq\sstar$.

Now, let $\cchange$ contain all clients with an active neighbor in $\schange$ before the insertion of $c$.  Since each client sends one unit of flow,
$\sum_{s \in \schange} \alphaold(s) \leq \sizeof{\cchange}$.
Now, because we had a balanced flow before the insertion of $c$ there
cannot be any edges in $G$ from $\cchange$ to $\sminus$ (any such edge would be from a client $u\in\cchange$ to a server $v\in\sminus$ with $\alphaold(v)\leq\alphastar<\alphaold(s)$ for $s\in\schange$ contradicting that $u$ had an active neighbor in $\schange$). Moreover,
in the balanced flow after the insertion of $c$, 
there are no active edges from $\cchange$ to $\splus$ (any such edge would be from a client $u\in\cchange$ to a server $v\in\splus$ with $\alphanew(v)>\alphastar=\alphanew(s)$ for all $s\in\schange$ so is not active).
Thus, all active edges incident to $\cchange$ go to $\schange$, so
$\sum_{s \in \schange} \alphanew(s) \geq \sizeof{\cchange}$. 
This contradicts the earlier fact that $\sum_{s \in \schange} \alphaold(s) \leq \sizeof{\cchange}$,
since by definition of $\schange$ we have
$\sum_{s \in \schange} \alphanew(s) < \sum_{s \in \schange} \alphaold(s)$. 
\end{proof}

The next lemma formalizes the following argument:
Say that we insert a new client $c$, and for simplicity say that $c$ is only incident to server
$s$. Now, $c$ will have no choice but to send all of its flow to $s$, but that does not imply
that $\Delta\alpha(s) = 1$, since other clients will balance by retracting their flow from
$s$ and sending it elsewhere. But by the assumption that the flow was balanced before the 
insertion of $c$, all this new flow can only flow ``upward'' from $s$: it cannot end up increasing the flow
on some $s^{-}$ with $\alphaold(s^{-}) < \alphaold(s)$. Along the same lines of intuition, even if $c$ has several neighbors, inserting $c$ cannot affect the flow of servers whose original flow was less than the lowest original flow among the neighbors of $s$.

\begin{lemma}\label{lem:onlyhigher}
  When a new client $c$ is added  , $\Delta\alpha(s)=0$ for all $s$ where $\alphaold(s)<\min_{v\in N(c)}\alphaold(v)$.
\end{lemma}
\begin{proof}
  Let us first consider the balanced flow \emph{before} the insertion of $c$.
  
  Let $\splus = \set{s \in S \cond \alphaold(s) \geq \min_{v\in N(c)}\alphaold(v)}$ and define
  $\sminus = S \setminus \splus$. We want to show that $\Delta\alpha(s)=0$  for all servers $s$ in $\sminus$.

Define $\cplus$ to be the set of client vertices 
whose neighbors are all in $\splus$; that is,
$\cplus = \set{c \in C \cond N(c) \subseteq \splus}$.
Note that the following holds before the insertion of $c$:
by definition of $\cplus$ there are no edges
in $G$ from $\cplus$ to $\sminus$, and
because the flow is balanced, 
there are no \emph{active} edges
from $\cminus$ to $\splus$.
Thus, $\sum_{s \in \sminus} \alphaold(s) = \sizeof{\cminus}$.


Now consider the insertion of $c$. 
By definition of $\sminus$ the new client $c$ has no neighbors in $\sminus$,
so it is still the case that only clients in $\cminus$ have neighbors in $\sminus$.
Thus, in the new balanced flow we still have have that 
$\sum_{s \in \sminus} \alphanew(s) \leq \sizeof{\cminus}$.
But this means that $\sum_{s \in \sminus} \Delta\alpha(s)\leq 0$, so if $\Delta\alpha(s_1)>0$ for some $s_1 \in \sminus$ then $\Delta\alpha(s_2)<0$ for some $s_2 \in \sminus$, which contradicts Lemma~\ref{lem:nondecreasing}.
\end{proof}

\section{Analyzing replacements in maximum matching}\label{sec:maxmatch}
We now consider how server flows relate to the length of augmenting paths.

\begin{lemma}\label{lem:alphalessone}
The graph $(C \cup S, E)$ contains a matching of size $\sizeof{C}$,
if and only if $\alpha(s) \leq 1$ for all $s\in S$.
\end{lemma}
\begin{proof}
Let $\alphamax=\max_{s\in S}\alpha(s)$.  It follows directly from Lemma~\ref{lem:maxalpha} that 
$\sizeof{K}\leq\sizeof{N(K)}$ for all $K\subseteq C$ if and only if $\alphamax\leq 1$.
The corollary then follows from Hall's Theorem (Theorem~\ref{thm:hall})
\end{proof}

It is possible that in the original graph $G = (C \cup S, E)$, there are 
many clients that cannot be matched. 
But recall that by Observation~\ref{obs:ignore-clients},
if a client cannot be matched when it is inserted, then it can be effectively 
ignored for the rest of the algorithm.
This motivates the following definition:

\begin{definition}
We define the set $\cm \subseteq C$ as follows. 
When a client $c$ is inserted, consider the set of clients $C'$
before $c$ is inserted: then $c \in \cm$
if the maximum matching in $(C' \cup \set{c} \cup S, E)$ is
greater than the maximum matching in $(C' \cup S, E)$.
Define $\gm = (\cm \cup S, E)$.
\end{definition}

\begin{observation}
\label{obs:cm}
When a client $c \in \cm$ is inserted the SAP algorithm finds an
augmenting path from $c$ to a free server; this follows from the fact that SAP always
maintains a maximum matching (Lemma~\ref{lem:sap-correctness}).
By Observation~\ref{obs:ignore-clients}, if $c \notin \cm$ then no augmenting
path goes through $c$ during the entire sequence of insertions. By the same
observation, once a vertex $c \in \cm$ is inserted
it remains matched through the entire sequence of insertions. 
\end{observation}

\begin{definition}
\label{dfn:alpham}
Given any $s \in S$, 
let $\alpham(s)$ be the flow into $s$ in some balanced server flow in $\gm$;
by Lemma~\ref{lem:alphaunique} $\alpham(s)$ is uniquely defined.
\end{definition}

\begin{observation}
\label{obs:alpham}
By construction $\gm$ contains a matching of size $\sizeof{\cm}$,
so by Lemma~\ref{lem:alphalessone} $\alpham(s) \leq 1$ for all $s \in S$.
Finally, note that since $\cm \subseteq C$, we clearly have $\alpham(s) \leq \alpha(s)$
\end{observation}

\begin{definition}
  \label{dfn:augmenting-tail}
  Define an \emph{augmenting tail} from a vertex $v$ to be an alternating path that starts in $v$ and ends in an unmatched server. We call an augmenting tail \emph{active} if all the edges on the alternating path that are not in the matching are active. 
\end{definition}

Note that augmenting tails as defined above are an obvious extension of the concept of augmenting paths: 
Every augmenting path for a newly arrived client $c$ consists of an edge $(c,s)$, plus an augmenting tail from some server $s\in N(c)$.

We are now ready to prove our main lemma connecting the balanced server flow to augmenting paths.
We show that if some server $s$ has small $\alpha(s)$, then regardless of the particular matching at hand, there is guaranteed to be a \emph{short} active augmenting tail from $s$. 
Since every \emph{active} augmenting tail is by definition an augmenting tail, this implies that any newly newly inserted client $c$ that is incident to $s$ has a short augmenting path to an unmatched server.

\begin{lemma}[Expansion Lemma]\label{lem:expand}
  Let $s\in S$, and suppose
  $\alpham(s)=1-\epsilon$ for some $\epsilon>0$.  Then there is an
  active augmenting tail for $s$ of length at most $\frac{2}{\epsilon}\ln(\sizeof{\cm})$.
\end{lemma}
\begin{proof}
By our definition of active edges, it
is not hard to see that any server $s'$ reachable from $s$ by an active augmenting tail has 
$\alpham(s') \leq 1-\epsilon$.

 For $i\geq 1$, let $K_i$ be the set of clients $c$ such that there is an active augmenting tail $s_0, c_0,  \ldots, c_{k-1}, s_k$ from $s$ with $c=c_j$ for some $j<i$. Let $k_i=\sizeof{K_i}$. Note that $k_1=1$, $K_1\subseteq K_2\subseteq \ldots\subseteq K_i$, and

  \begin{align*}
    k_i = \sizeof{K_i}\leq \sum_{s'\in A(K_i)} \alpham(s') \leq \sum_{s'\in A(K_i)}(1-\epsilon) = \sizeof{A(K_i)}(1-\epsilon)
  \end{align*}
  Thus
  \begin{align*}
    \sizeof{A(K_i)} \geq \frac{k_i}{1-\epsilon}
  \end{align*}

Suppose there is no active augmenting tail from $s$ of length $\leq 2(i-1)$, then every server in $A(K_i)$ is matched, and the clients they are matched to are exactly $K_{i+1}$.  There is a bijection between $A(K_i)$ and $K_{i+1}$ given by the perfect matching, so we have $k_{i+1} = \sizeof{A(K_i)}$ and thus $\sizeof{\cm} \geq k_{i+1} \geq \frac{1}{1-\epsilon}k_i \geq (\frac{1}{1-\epsilon})^ik_1 = (\frac{1}{1-\epsilon})^i$.  It follows that $i\leq \frac{\ln\sizeof{\cm}}{\ln\frac{1}{1-\epsilon}}\leq\frac{1}{\epsilon}\ln\sizeof{\cm}$, where the last inequality follows from $1-\epsilon\leq e^{-\epsilon}$.  Thus for any $i>\frac{1}{\epsilon}\ln\sizeof{\cm}$ there exists an active augmenting tail of length at most $2(i-1)$, and the result follows.
\end{proof}

We are now able to prove the key lemma of our paper, which we showed
in Section~\ref{sec:overview} implies Theorem~\ref{thm:main-uncapacitated}.

\longpaths*

\begin{proof}
Recall that $n = \sizeof{C} \geq \sizeof{\cm}$.
The lemma clearly holds for $h \leq 4\ln(n)$ because there at most $n$ augmenting paths in total.
We can thus assume for the rest of the proof that $h > 4\ln(n)$.
Recall by Observation~\ref{obs:cm} that any augmenting path is contained entirely in $\gm$.
Now, let $\cstar \subseteq \cm$ be the set of clients whose shortest augmenting path have length at least $h+1$ when they are added.  Our goal is to show that $\sizeof{\cstar}\leq4n\ln(n)/h$. For each $c\in\cstar$ the shortest augmenting
tail from each server $s \in N(c)$ has length at least $h$ and so by the Expansion 
Lemma~\ref{lem:expand}, 
each server $s \in N(c)$ has  $\alpham(s) \geq 1 - 2\ln(n)/h$.
Let $\sstar$ be the set of all servers that at some point have
$\alpham(s) \geq 1 - 2\ln(n)/h$; by Lemma~\ref{lem:nondecreasing},
this is exactly the set of servers $s$ such that $\alpham(s) \geq 1 - 2\ln(n)/h$
after all clients have been inserted. 
By Lemma~\ref{lem:onlyhigher}, if $c\in\cstar$, the insertion of $c$ only increases the flow on servers in $\sstar$ that already had flow
at least $1 - 2\ln(n)/h$. Since by Observation~\ref{obs:alpham} 
$\alpham(s) \leq 1$ for all $s \in S$, the flow of each server in $\sstar$ 
can only increase by at most $2\ln(n)/h$. But then, since the client $c$ contributes with exactly one unit of flow,
the total number of such clients is  $\sizeof{\cstar}\leq(2\log(n)/h)\sizeof{\sstar}$.
We complete the proof by showing that $\sizeof{\sstar} < 2n$.
This follows from the fact that each client $c \in \cm$ sends one unit of flow,
so $n \geq \sizeof{\cm} \geq (1 - 2\ln(n)/h)\sizeof{\sstar} > \sizeof{\sstar}/2$,
where the last inequality follows from the assumption that $h > 4\ln(n)$.
\end{proof}

\section{Implementation}\label{sec:implementation}
In the previous section we proved that augmenting along a shortest augmenting path yields a total of $\Oo(n\log^2 n)$ replacements. But the naive implementation would spend $\Oo(m)$ time per inserted vertex, leading to 
total time $\Oo(mn)$ for actually maintaining the matching. In this section, we show how to find the augmenting paths more quickly, and thus maintain the optimal matching at all times in $\Oo(m\sqrt{n}\sqrt{\log n})$ total time, differing only by an $\Oo(\sqrt{\log n})$ factor from the classic offline algorithm of Hopcroft and Karp algorithm for static graphs~\cite{Hopcroft73}.

\begin{definition}
Define the height of a vertex $v$ (server or client) to be the length of
the shortest augmenting tail (Definition~\ref{dfn:augmenting-tail})
from $v$. If no augmenting tail exists, we set the height to $2n$.
\end{definition}

At a high level, our algorithm is very similar to the standard $\Oo(m\sqrt{n})$ blocking flow algorithm.
We will keep track of heights to find shortest augmenting paths of length at most $\sqrt{n}\sqrt{\log n}$.
We will find longer augmenting paths using the trivial $\Oo(m)$ algorithm,
and use Lemma~\ref{lem:bound-long-paths} to bound the number of such paths.
Our analysis will also require the following lemma:

\begin{lemma}
\label{lem:no-augmenting-path}
For any server $s \in S$, there is an augmenting tail from $s$ to an unmatched server
if and only if $\alpham(s) < 1$.
\end{lemma}

\begin{proof}
If $\alpham(s) < 1$, then the existence of \emph{some} tail follows directly from the Expansion Lemma~\ref{lem:expand}.
Now let us consider $\alpham(s) = 1$.
Let $S_1 = \set{s \in S \cond \alpham(s) = 1}$.
Since $1$ is the maximum possible value of $\alpham(s)$ (Observation~\ref{obs:alpham}),
Lemma~\ref{lem:alphaset} implies that there is a set of clients $C_1 \in \cm$ such that 
$N(C_1) = S_1$ and $\sizeof{C_1} = \sizeof{S_1}$. Now since every client in $C_1$ is matched,
every server $S_1$ is matched to some client in $C_1$. Every augmenting tail from some $s \in S_1$
must start with a matched edge, so it must go through $C_1$, so it never reaches a server
outside of $N(C_1) = S_1$, so it can never reach a free server.
\end{proof}

We now turn to our implementation of the SAP protocol. We will use a dynamic single-source shortest paths
algorithm as a building block. We start by defining a directed graph $D$ such that
maintaining distances in $D$ will allow us to easily find shortest augmenting paths as new clients
are inserted.

Let $D$ be the directed graph obtained from $G=(C\cup S,E)$ by directing all unmatched edges from $C$ to $S$, and all matched edges from $S$ to $C$, and finally adding a \emph{sink} $t$ with an edge from all unmatched vertices, as well as edge from all clients in $C$ that have not yet arrived.
Any alternating path in $G$ corresponds to a directed path in $D\setminus\set{t}$ and vice-versa. In particular, it is easy to see that if $P$ is a shortest path in $D$ from a matched server $s$ to the sink $t$, then $P\setminus\set{t}$ is a shortest augmenting tail from $s$ to a free server. Similarly,
for any client $c$ that has arrived (so edge $(c,t)$ is deleted) but is not yet matched,
if $P$ is the shortest path from $c$ to $t$ in $D$,
then $P\setminus\set{t}$ is a shortest augmenting path for $c$ in $G$. Furthermore, augmenting down this path in $G$ corresponds (in $D$) to changing the direction of all edges on $P\setminus\set{t}$ and deleting the edge on $P$ incident to $t$.

We can thus keep track of shortest augmenting paths by using a simple dynamic shortest path algorithm
to maintain shortest paths to $t$ in the changing graph $D$.
We will use a modification of Even and Shiloach (See~\cite{Shiloach:1981}) to maintain a shortest path tree in $D$ to $t$ from all vertices of height at most $h=\sqrt{n}\sqrt{\log n}$.  The original version by Even and Shiloach worked only for undirected graphs, and only in the decremental setting where 
the graph only undergoes edge deletions, never edge insertions.  This was later extended by King~\cite{King99} to work for directed graphs. The only-deletions setting
is too constrained for our purposes
because we will need to insert edges into $D$; augmenting down a path $P$ corresponds to 
deleting the edges on $P$ and inserting the reverse edges.
Fortunately, it is well known that the Even and Shiloach tree can be extended to the setting
where there are both deletions and insertions, as long as the latter are guaranteed not to
decrease distances; we will show that this in fact applies to our setting.

\begin{lemma}[{Folklore. See e.g.\cite{BernsteinR11,BernsteinC16}}]\label{lem:ES-extended}
  Let $G = (V, E)$ be a dynamic directed or undirected graph with positive integer weights, let $t$ be a fixed sink, and say that for every vertex $v$ we are guaranteed that the distance $\dist(v,t)$ never decreases due to an edge insertion. Then we can maintain a tree of shortest paths to $t$ up to distance $d$ in total time $\Oo(m\cdot d+\Delta)$, where $m$ is the total number of edges $(u,v)$ such 
that $(u,v)$ is in the graph at any point during the update sequence, and $\Delta$ is the total number of edge changes.
\end{lemma}

\implement*

\begin{proof}
  We will explicitly maintain the graph $D$, and use the extended Even-Shiloach tree structure from Lemma~\ref{lem:ES-extended} to maintain a tree $T$ of shortest paths to $t$ up to distance $h=\sqrt{n}\sqrt{\log n}$. Every vertex will either be in this tree (and hence have height less than $h$), or be marked as a \emph{high} vertex.
  When a new client $c$ arrives, we update $D$ (and $T$) by first adding edges to $N(c)$ from $c$, and then deleting the dummy edge from $c$ to $t$.  Note that because the deletion of edge $(c,t)$ comes last,
	the inserted edges do not change any distances to $t$.
	 We then use $D$ and $T$ to find a shortest augmenting path. We consider two cases.

  The first case is when $c$ is not high.  Then $T$ contains a shortest path $P$ from $c$ to $t$.

  The second case is when $c$ is high. In this case we can just brute-force search for a shortest path $P$ from $c$ to $t$ in time $\Oo(m)$. \emph{If we do not find a path from $c$ to $t$, then we remove all servers and clients encountered during the search, and continue the algorithm in the graph with these vertices removed}.

  In either case, if a shortest path $P$ from $c$ to $t$ is found, we augment down $P$ and
	then make the corresponding changes to $D$:
	we first reverse the edges on $P\setminus\set{t}$ in order starting with the edge closest to $c$, and then we delete the edge $(s,t)$ on $P$ incident to $t$ (because the server $s$ is now matched).  Each edge reversal is done by first inserting the reversed edge, and then deleting the original. Note that since $P$ is a shortest path, none of these edge insertions change the distances to $t$. 
\paragraph{Correctness:}
We want to show that our implementation chooses a shortest augmenting path at every step. This is clearly true if we always find an augmenting path, but otherwise becomes a bit more subtle as we delete vertices from the graph after a failed brute-force search. 
We must thus show that any vertex deleted in this way cannot have participated in any future
augmenting path.

To see this, note that when our implementation deletes a server $s \in S$, there must have been 
no augmenting
path through $s$ at the time that $s$ was deleted. By Lemma~\ref{lem:no-augmenting-path}, this implies
that $\alpham(s) = 1$. But then by Lemma~\ref{lem:nondecreasing} we have $\alpham(s) = 1$
for all future client insertions as well. 
(Recall that by Observation~\ref{obs:alpham} we never have $\alpham(s) > 1$.)
Thus by Lemma~\ref{lem:no-augmenting-path} there is never an augmenting path through $s$
after this point, so $s$ can safely be deleted from the graph. Similarly,
if a client $c$ is deleted from the graph, then all of its neighboring servers had no augmenting
tails at that time, so they all have $\alpham(s) = 1$, so there will never be an 
augmenting path through $c$.

\paragraph{Running time:} 

There are three factors to consider.
\begin{enumerate}
    \item the time to follow the augmenting paths and maintain $D$.\label{item:follow}
    \item the time to maintain $T$.\label{item:maintain}
    \item the time to brute-force search for augmenting paths.\label{item:brute}
\end{enumerate}

Item~\ref{item:follow} takes $\Oo(m+n\log^2 n)$ time because we need $\Oo(1)$ time to add each of the $m$ edges and to follow and reverse each edge in the augmenting paths, and by Theorem~\ref{thm:main-uncapacitated} the total length of augmenting paths is $\Oo(n\log^2 n)$.

For Item~\ref{item:maintain}, it is easy to see that the total number of 
edges ever to appear in $D$ is $m = O(|E|)$;
$D$ consists only of dummy edges to the sink $t$, and edges in the original graph
oriented in one of two directions.
By Item~\ref{item:follow}, the number of changes to $D$ is $\Oo(m + n\log^2 n)$.
Thus by Lemma~\ref{lem:ES-extended} the total time to maintain $T$ is 
 $\Oo(mh + n\log^2 n)$ .

For Item~\ref{item:brute} we consider two cases. The first is brute-force searches
which result in finding an augmenting path. These take a total of
$\Oo(mn\log(n)/h)$ time
because by Lemma~\ref{lem:bound-long-paths} during the course of the entire algorithm
there are at most $\Oo(n\log(n)/h)$ augmenting paths of length $\geq h$, and each such path  
requires $\Oo(m)$ time to find. The second case to consider is brute-force searches
that do not result in an augmenting path. These take total time $\Oo(m)$
because once a vertex participates in such a search, it is deleted from the graph with all its incident edges.

Summing up, the total time used is $\Oo(mh+n\log^2 n+mn\log(n)/h+m)$, which for our choice of $h=\sqrt{n}\sqrt{\log n}$ is $O(m\sqrt{n}\sqrt{\log n})$.
\end{proof}

\section{Extensions}\label{sec:extend}

In many applications of online bipartite assignments, it is natural to consider the extension in which each server can serve multiple clients. Recall from the introduction that
we examine two variants: capacitated assignment, where each server comes with a fixed capacity which we are not allowed to exceed, and minimizing maximum server load, in which there is no upper limit to the server capacity, but we wish to minimize the maximum number of clients served by any server.
We show that there is a substantial difference between the number of reassignments: Capacitated assignment is equivalent to uncapacitated online matching with replacements, but for minimizing maximum load, we show a significantly higher lower bound.

\subsection{Capacitated assignment}
\label{sec:capacitated-assignment}
We first consider the version of the problem where each server can be matched to multiple clients.
Each server comes with a positive integer capacity $u(s)$, which denotes how many clients can
be matched to that server. The greedy algorithm is the same as before: when a new client is inserted,
find the shortest augmenting path to a server $s$ that currently has less than $u(s)$ 
clients assigned.

\capacitate*

\begin{proof}
There is a trivial reduction from any instance of capacitated assignment to one of uncapacitated matching
where each server can only be matched to one client: simple create $u(s)$ copies of each server $s$.
This reduction was previously used in~\cite{BernsteinKPPS17}.
When a client $c$ is inserted, if there is an edge $(c,s)$ in the original graph, 
then add edges from $c$ to every copy of $s$. It is easy to see that the number of flips
made by the greedy algorithm in the capacitated graph is exactly
equal to the number made in the uncapacitated graph, which by Theorem
\ref{thm:main-uncapacitated} is $\Oo(n\log^2 n)$. (Note that although the constructed uncapacitated
graph has more servers than the original capacitated graph, the number of clients $n$
is exactly the same in both graphs.)
\end{proof}

\subsection{Minimizing maximum server load}
\label{sec:MinMax}

In this section, we analyze the online assignment problem. Here, servers may have an unlimited load, but we wish to minimize maximum server load.  
\begin{definition}
\label{def:assignment}
Given a bipartite graph $G=(C\cup S,E)$, an assignment $\Aa : C \rightarrow S$ assigns 
each client $c$ to a server $\Aa(c) \in S$. 
Given some assignment $\Aa$, for any $s \in S$ let the \emph{load} of $s$, 
denoted $\load_{\Aa}(s)$, be
the number of clients assigned to $s$; when the assignment $\Aa$ is clear
from context we just write $\load(s)$.
Let $\load(\Aa) = \max_{s \in S} \load_{\Aa}(s)$.
Let $\opt(G)$ be the minimum load among all possible assignments from $C$ to $S$.
\end{definition}

In the online assignment problem, clients are again inserted one by one with all 
their incident edges, and the goal is to maintain an assignment with minimum possible load.
More formally, define $G_t=(C_t\cup S,E_t)$ to be the graph after exactly $t$ clients have arrived,
and let $\Aa_t$ be the assignment at time $t$. Then we must have that for all $t$,
$\load(\Aa_t) = \opt(G_t)$. The goal is to make as few changes to the assignment as possible. 

\cite{Gupta:2014} and~\cite{BernsteinKPPS17} showed how to solve this problem with approximation: namely, with
only $\Oo(1)$ amortized changes per client insertion they can maintain
an assignment $\Aa$ such that for all $t$, $\load(\Aa_t) \leq 8\opt(G_t)$. 
Maintaining an approximate assignment is thus not much harder than maintaining an approximate maximum matching, so one might have hoped
that the same analogy holds for the exact case, and that it is possible
to maintain an optimal assignment with amortized $\Oo(\log^2 n)$ 
changes per client insertion. We now present a lower bound
disproving the existence of such an upper bound. The lower bound is not specific to the greedy algorithm,
and applies to any algorithm for maintaining an assignment $\Aa$ of minimal load.
In fact, the lower bound applies even if the algorithm knows the entire graph $G$
in advance; by contrast, if the goal is only to maintain a maximum matching,
then knowing $G$ in advance trivially leads to an online matching algorithm that never has to rematch any vertex.

\lowerbound*

The main ingredient of the proof is the following lemma:

\begin{lemma}
\label{lem:assignment-lower-bound}
For any positive integer $\maxload$ divisible by $4$, there exists a graph $G=(C\cup S,E)$ along with
an ordering in which clients in $C$ are inserted, such that
$\sizeof{C} = \maxload^2$, $\sizeof{S} = \maxload$, $\opt(G) = \maxload$,
and any algorithm for maintaining an optimal assignment $\Aa$ requires $\Omega(\maxload^3)$
changes to $\Aa$.
\end{lemma}

\begin{proof}
Let $S = \set{s_1, s_2, ..., s_\maxload}$.
We partition the clients in $C$ into $\maxload$ blocks $C_1, C_2, ..., C_\maxload$,
where all the clients in a block have the same neighborhood. In particular, 
clients in $C_\maxload$ only have a single edge to server $s_\maxload$,
and clients in $C_i$ for $i < \maxload$ have an edge to $s_i$ and $s_{i+1}$.

The online sequence of client insertions begins by adding $\maxload/2$ clients
to each block $C_i$. The online sequence then proceeds to alternate between \emph{down-heavy} epochs and 
\emph{up-heavy} epochs, where a down-heavy epoch inserts 2 clients into blocks 
$C_1, C_2, ..., C_{\maxload/2}$ (in any order), while an up-heavy epoch inserts 2 clients into blocks
$C_{\maxload/2 + 1}, ..., C_\maxload$. The sequence then terminates after $\maxload/2$ 
such epochs: $\maxload/4$ up-heavy ones and $\maxload/4$ down-heavy ones in alternation. 
Note that a down-heavy epoch followed by an up-heavy one simply adds
two clients to each block. 
Thus the final graph has $\sizeof{C_i} = \maxload$
for each $i$, so the graph $G=(C\cup S,E)$ satisfies the desired conditions that $\sizeof{C} = \maxload^2$
and $\opt(G) = \maxload$.

\begin{figure}[h!]
    \begin{tikzpicture*}{\textwidth}
        \begin{scope}[
            vertex/.style={
                draw,
                circle,
                minimum size=2mm,
                inner sep=0pt,
                outer sep=0pt            },
            vedge/.style={
                near end,
                below,
                rotate=90,
                font=\tiny,
            },
            dedge/.style={
                near start,
                below,
                rotate=63.43,
                font=\tiny,
            },
            every label/.append style={
                rectangle,
                                font=\tiny,
            }
            ]
            \node[vertex,label={below:$C_1$}] (c1) at (1,0) {};
            \node[vertex,label={below:$C_2$}] (c2) at (2,0) {};
            \node[vertex,label={below:$C_3$}] (c3) at (3,0) {};
            \node[vertex,label={below:$C_{L/2}$}] (cL2) at (6,0) {};
            \node[vertex,label={below:$C_{L/2+1}$}] (cL2p1) at (7,0) {};
            \node[vertex,label={below:$C_{L-3}$}] (cLm3) at (10,0) {};
            \node[vertex,label={below:$C_{L-2}$}] (cLm2) at (11,0) {};
            \node[vertex,label={below:$C_{L-1}$}] (cLm1) at (12,0) {};
            \node[vertex,label={below:$C_L$}] (cLm0) at (13,0) {};
            
            \node[vertex,label={above:$s_1$}] (s1) at (1,2) {};
            \node[vertex,label={above:$s_2$}] (s2) at (2,2) {};
            \node[vertex,label={above:$s_3$}] (s3) at (3,2) {};
            \node[vertex,label={above:$s_4$}] (s4) at (4,2) {};
            \node[vertex,label={above:$s_{L/2}$}] (sL2) at (6,2) {};
            \node[vertex,label={above:$s_{L/2+1}$}] (sL2p1) at (7,2) {};
            \node[vertex,label={above:$s_{L/2+2}$}] (sL2p2) at (8,2) {};
            \node[vertex,label={above:$s_{L-2}$}] (sLm2) at (11,2) {};
            \node[vertex,label={above:$s_{L-1}$}] (sLm1) at (12,2) {};
            \node[vertex,label={above:$s_L$}] (sLm0) at (13,2) {};
            
            \draw[dotted] (6.5,-0.75) -- (6.5,2.75);
            
            \draw (c1) -- (s1) node[vedge] {$\beta$};
            \draw (c1) -- (s2) node[dedge] {$0$};
            \draw (c2) -- (s2) node[vedge] {$\beta$};
            \draw (c2) -- (s3) node[dedge] {$0$};
            \draw (c3) -- (s3) node[vedge] {$\beta$};
            \draw (c3) -- (s4) node[dedge] {$0$};
            \draw[dashed] ($(s4.south)+(0,-2)$) -- (s4);
            
            \node at (4.75,1) {$\cdots$};
            
            \draw[dashed] ($(sL2.south)+(-1,-2)$) -- (sL2);
            \draw (cL2) -- (sL2) node[vedge] {$\beta$};
            \draw (cL2) -- (sL2p1) node[dedge] {$0$};
            \draw (cL2p1) -- (sL2p1) node[vedge] {$\beta$};
            \draw (cL2p1) -- (sL2p2) node[dedge] {$0$};
            \draw[dashed] ($(sL2p2.south)+(0,-2)$) -- (sL2p2);
            
            \node at (9,1) {$\cdots$};
            
            \draw[dashed] (cLm3) -- ($(cLm3.north)+(0,2)$);
            \draw (cLm3) -- (sLm2) node[dedge] {$0$};
            \draw (cLm2) -- (sLm2) node[vedge] {$\beta$};
            \draw (cLm2) -- (sLm1) node[dedge] {$0$};
            \draw (cLm1) -- (sLm1) node[vedge] {$\beta$};
            \draw (cLm1) -- (sLm0) node[dedge] {$0$};
            \draw (cLm0) -- (sLm0) node[vedge] {$\beta$};
        \end{scope}
    \end{tikzpicture*}
    \vspace{-0.25cm}
    \caption{Number of assignments of each type after first $L/2$ clients added to each block, and after each up-heavy phase. Each $C_i$ has $\beta$ clients. Each server has $\beta$ clients assigned.}
    \label{fig:assignment-lower-upheavy}
\end{figure}
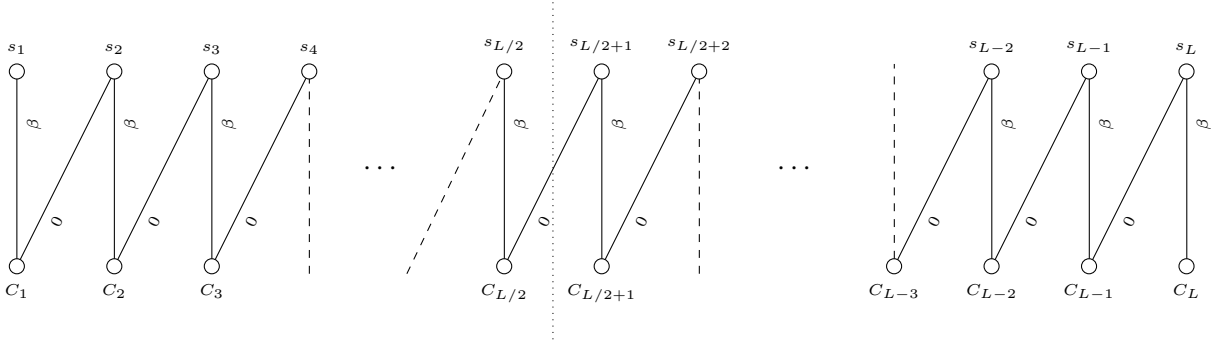

We complete the proof by showing that all the client insertions during a single
down-heavy epoch cause the algorithm to make at least $\Omega(\maxload^2)$ changes
to the assignment; the same analysis applies to the up-heavy epochs as well. 
Consider the $k$th down-heavy epoch of client insertions. Let
$\beta = \maxload/2 + 2(k-1)$ and consider the graph $\Gold=(\cold\cup S,\eold)$ before the down-heavy epoch:
it is easy to see that  every block $C_i$ has exactly $\beta$ clients, 
that $\opt(\Gold) = \beta$, 
and that there is exactly one assignment $\Aaold$ that adheres to this maximum load: $\Aaold$
assigns all clients in block $C_i$ to server $s_i$ (see Figure~\ref{fig:assignment-lower-upheavy}).

Now, consider the graph $\Gnew=(\cnew \cup S,\enew)$ after the down-heavy epoch. Blocks $C_1, C_2, ..., C_{\maxload/2}$
now have $\beta + 2$ clients, while blocks $C_{\maxload/2 + 1},...,C_{\maxload}$ still only have $\beta$.
We now show that $\opt(\Gnew) = \beta + 1$. In particular, 
recall that $\beta \geq \maxload/2$ and consider
the following assignment $\Aanew$: for $i \leq \maxload/2$, $\Aanew$
assigns $\beta + 2 - i \geq 2$ clients from $C_i$ 
to $s_i$ and $i$ clients in $C_i$ to $s_{i+1}$;
for $\maxload/2 < i \leq \maxload$, $\Aanew$ assigns $\beta + i - \maxload \geq 0$ 
clients in $C_i$ to $s_i$,
and $\maxload - i$ clients from $C_i$ to $s_{i+1}$. 
(In particular, all $\beta$ clients in $C_{\maxload}$ are assigned to $s_\maxload$,
which is necessary as there is no server $s_{\maxload+1}$). 
It is easy to check that for every $s \in S$, 
$\load_{\Aanew}(s) = \beta + 1$ (see Figure~\ref{fig:assignment-lower-downheavy}).

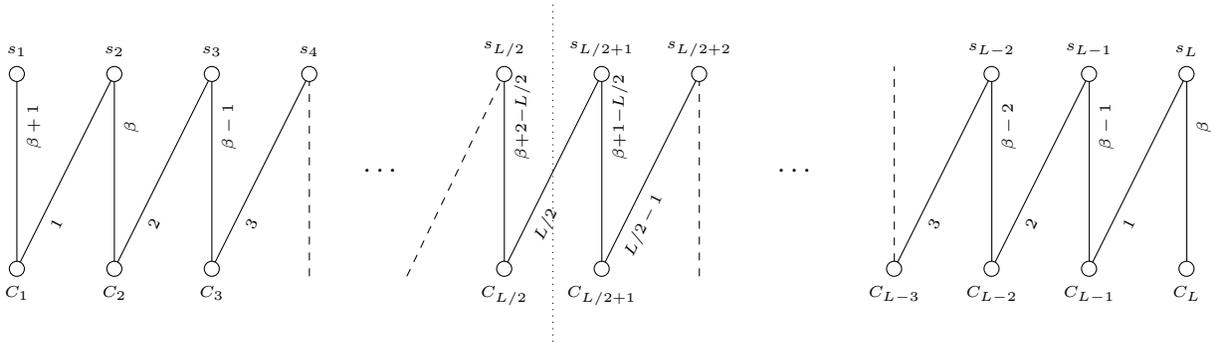
\begin{figure}[h!]
    \begin{tikzpicture*}{\textwidth}
        \begin{scope}[
            vertex/.style={
                draw,
                circle,
                minimum size=2mm,
                inner sep=0pt,
                outer sep=0pt            },
            vedge/.style={
                near end,
                below,
                rotate=90,
                font=\tiny,
            },
            dedge/.style={
                near start,
                below,
                rotate=63.43,
                font=\tiny,
            },
            every label/.append style={
                rectangle,
                                font=\tiny,
            }
            ]
            \node[vertex,label={below:$C_1$}] (c1) at (1,0) {};
            \node[vertex,label={below:$C_2$}] (c2) at (2,0) {};
            \node[vertex,label={below:$C_3$}] (c3) at (3,0) {};
            \node[vertex,label={below:$C_{L/2}$}] (cL2) at (6,0) {};
            \node[vertex,label={below:$C_{L/2+1}$}] (cL2p1) at (7,0) {};
            \node[vertex,label={below:$C_{L-3}$}] (cLm3) at (10,0) {};
            \node[vertex,label={below:$C_{L-2}$}] (cLm2) at (11,0) {};
            \node[vertex,label={below:$C_{L-1}$}] (cLm1) at (12,0) {};
            \node[vertex,label={below:$C_L$}] (cLm0) at (13,0) {};
            
            \node[vertex,label={above:$s_1$}] (s1) at (1,2) {};
            \node[vertex,label={above:$s_2$}] (s2) at (2,2) {};
            \node[vertex,label={above:$s_3$}] (s3) at (3,2) {};
            \node[vertex,label={above:$s_4$}] (s4) at (4,2) {};
            \node[vertex,label={above:$s_{L/2}$}] (sL2) at (6,2) {};
            \node[vertex,label={above:$s_{L/2+1}$}] (sL2p1) at (7,2) {};
            \node[vertex,label={above:$s_{L/2+2}$}] (sL2p2) at (8,2) {};
            \node[vertex,label={above:$s_{L-2}$}] (sLm2) at (11,2) {};
            \node[vertex,label={above:$s_{L-1}$}] (sLm1) at (12,2) {};
            \node[vertex,label={above:$s_L$}] (sLm0) at (13,2) {};
            
            \draw[dotted] (6.5,-0.75) -- (6.5,2.75);
            
            \draw (c1) -- (s1) node[vedge] {$\beta+1$};
            \draw (c1) -- (s2) node[dedge] {$1$};
            \draw (c2) -- (s2) node[vedge] {$\beta$};
            \draw (c2) -- (s3) node[dedge] {$2$};
            \draw (c3) -- (s3) node[vedge] {$\beta-1$};
            \draw (c3) -- (s4) node[dedge] {$3$};
            \draw[dashed] ($(s4.south)+(0,-2)$) -- (s4);
            
            \node at (4.75,1) {$\cdots$};
            
            \draw[dashed] ($(sL2.south)+(-1,-2)$) -- (sL2);
            \draw (cL2) -- (sL2) node[vedge] {$\quad\beta{+}2{-}L{/}2$};
            \draw (cL2) -- (sL2p1) node[dedge] {$L/2$};
            \draw (cL2p1) -- (sL2p1) node[vedge] {$\quad\beta{+}1{-}L{/}2$};
            \draw (cL2p1) -- (sL2p2) node[dedge] {$L/2-1$};
            \draw[dashed] ($(sL2p2.south)+(0,-2)$) -- (sL2p2);
            
            \node at (9,1) {$\cdots$};
            
            \draw[dashed] (cLm3) -- ($(cLm3.north)+(0,2)$);
            \draw (cLm3) -- (sLm2) node[dedge] {$3$};
            \draw (cLm2) -- (sLm2) node[vedge] {$\beta-2$};
            \draw (cLm2) -- (sLm1) node[dedge] {$2$};
            \draw (cLm1) -- (sLm1) node[vedge] {$\beta-1$};
            \draw (cLm1) -- (sLm0) node[dedge] {$1$};
            \draw (cLm0) -- (sLm0) node[vedge] {$\beta$};
        \end{scope}
    \end{tikzpicture*}
    \vspace{-0.25cm}
    \caption{Number of assignments of each type after each down-heavy phase.
        Each $C_i$ has $\beta+2$ clients for $1\leq i\leq L/2$ and $\beta$ clients for $L/2+1\leq i\leq L$.  Each server has $\beta+1$ clients assigned.}
    \label{fig:assignment-lower-downheavy}
\end{figure}

We now argue that $\Aanew$ is in fact the only assignment
in $\Gnew$ with $\load(\Aanew) = \beta + 1$. 
Consider any assignment $\Aa$ for $\cnew$ with $\load(\Aa) = \beta + 1$.
Observe that since the
total number of clients in $\cnew$ is exactly $(\beta+1) \maxload$, we
must have that every server $s \in S$ has $\load(s) = \beta+1$ in $\Aa$. 
We now argue by induction that for $i \leq \beta/2$,
$\Aa$ assigns assigns $\beta + 2 - i$ clients from $C_i$ 
to $s_i$ and $i$ clients in $C_i$ to $s_{i+1}$ (exactly as $\Aanew$ does). 
The claim holds for $i=1$ because the only way $s_1$ can end up
with load $\beta + 1$ is if $\beta + 1$ clients from $C_1$ are assigned to it.
Now say the claim is true for some $i < \beta/2$. By the induction
hypothesis, $s_{i+1}$ has $i$ clients from $C_i$ assigned to it. 
Since $s_{i+1}$ must have total load $\beta + 1$, and
all clients assigned to it come from $C_i$ or $C_{i+1}$, $s_{i+1}$ must 
have $\beta + 1 - i = \beta + 2 - (i+1)$ clients assigned to it from $C_{i+1}$.

We now prove by induction that for all $\maxload/2 < i \leq \maxload$, 
$\Aa$ assigns $\beta + i - \maxload$ 
clients in $C_i$ to $s_i$,
and $\maxload - i$ clients from $C_i$ to $s_{i+1}$,
which proves that $\Aa = \Aanew$. 
The claim holds for $i = \maxload/2 + 1$ because 
we have already shown that in the above paragraph that $\maxload/2$ clients
assigned to $s_i = s_{\maxload/2 + 1}$ come from $C_{\maxload/2}$, so since
$\load(s_i) = \beta + 1$,
it must have $\beta + 1 - \maxload/2 = \beta + i -\maxload$
clients from $C_i$ assigned to it. Now, say that the claim is true
for some $i > \maxload/2$. Then by the induction step $s_{i+1}$
has $\maxload - i$ clients assigned to it from $C_i$,
so since $\load(s_{i+1}) = \beta + 1$, it has $\beta + (i+1)-\maxload$
clients assigned to it from $C_{i+1}$, as desired. The remaining
$\maxload - (i+1)$ clients in $C_{i+1}$ must then be assigned to $s_{i+2}$. 

We have thus shown that the online assignment algorithm
is forced to have assignment $\Aaold$ before the down-heavy epoch, and
assignment $\Aanew$ afterwards.
We now consider how many changes the algorithm must make to go from one to another.
Consider block $C_i$ for some $\maxload/2 < i \leq \maxload$. Note that because
the epoch of client insertions was down-heavy, 
$\sizeof{C_i} = \beta$ before and after the epoch. 
Now, in $\Aaold$ all of the clients in $C_i$ are matched to $s_i$. 
But in $\Aanew$, $\maxload-i$ of them are matched to $s_{i+i}$.
Thus, the total number of reassignments to get from $\Aaold$ to $\Aanew$
is at least $\sum_{\maxload/2 < i \leq \maxload} (\maxload - i) = \Omega(\maxload^2)$.
Since there are $\maxload/4$ down-heavy epochs, the total number of reassignments
over the entire sequence of client insertions is $\Omega(\maxload^3)$. 
\end{proof}

\begin{proof}[Proof of Theorem~\ref{thm:assignment-lower-bound}]
  Recall the assumption of the Theorem that $n/2 \geq \maxload^2$    .
  Simply let the graph $G$ consist of $\floor{n/\maxload^2}$ separate
  instances of the graph in Lemma~\ref{lem:assignment-lower-bound}, together with sufficient copies of $K_{1,1}$ to make the total number of clients $n$.
  The algorithm will have to make $\Omega(\maxload^3)$ changes in each such instance,
  leading to $\Omega(\maxload^3 \floor{n/\maxload^2}) = \Omega(n\maxload)$ changes in total.
\end{proof}

We now show a nearly matching upper bound which is off by a $\log^2 n$ factor.
As with the case of matching, this upper bound is achieved by the most natural SAP algorithm,
which we now define in this setting.
Since $\opt(G)$ may change as clients are inserted into $C$, whenever a new client is inserted,
the greedy algorithm must first compute $\opt(G)$ for the next client set. 
Note that the algorithm does not do any reassignments at this stage, it simply
figures out what the max load should be. $\opt(G)$ can easily
be computed in polynomial time: for example we could just compute the maximum matching
when every server has capacity $b$ for every $b = 1, 2, ..., \sizeof{C}$,
and then $\opt(G)$ is the minimum $b$ 
for which every client in $C$ is matched; for a more efficient approach see~\cite{BernsteinKPPS17}.
Now, when a new client $c$ is inserted, the algorithm first checks if $\opt(G)$
increases. If yes, the maximum allowable load on each server increases by $1$ so 
$c$ can just be matched to an arbitrary neighbor. Otherwise, SAP finds the shortest alternating path from $c$ to a server $s$ with $\load(s) < \opt(G)$:
an augmenting path is defined exactly the same way as in Definition~\ref{dfn:augmenting-path},
though there may now be multiple matching edges incident to every server. 
The proof of the upper bound will rely on the following very simple observation:

\begin{observation}
\label{obs:initial-matching}
For the uncapacitated problem of online maximum matching with replacements,
let us say that instead of starting with $C = \emptyset$, the algorithm
starts with some initial set of clients $C_0 \subset C$ already inserted, and an initial matching between $C_0$ and $S$. 
Then the total number of replacement made during all future client insertions is still 
upper bounded by the same 
$\Oo(n\log^2 n)$ as in
Theorem~\ref{thm:main-uncapacitated}, where $n$ is the number of clients in the final graph
(so $n$ is $\sizeof{C_0}$ plus the number of clients inserted).
\end{observation}

\begin{proof}
Intuitively, we could simply let our protocol start by unmatching all the clients in $C_0$,
and then rematching them according the SAP protocol, which would lead to $\Oo(n\log^2 n)$
replacements. In fact this initial unmatching is not actually necessary. Recall that the proof
of Theorem~\ref{thm:main-uncapacitated} follows directly from the key Lemma~\ref{lem:bound-long-paths},
which in term follows from the expansion argument in Lemma~\ref{lem:expand}. 
The expansion argument only refers to server necessities, not to the
particular matching maintained by the algorithm, so it will hold no matter what initial matching
we start with.
\end{proof}

\minmax*

\begin{proof}
Let us define epoch $i$ to contain all clients $c$ such that after the insertion of $c$ we have $\opt(G) = i$. We now define $n_i$ as the total number of clients added by the end of epoch $i$
(so $n_i$ counts clients from previous epochs as well).
Extend the reduction in the proof of Theorem~\ref{thm:main-capacitated} from~\cite{BernsteinKPPS17} as follows: between any two epochs, add a new copy of each server, along with all of its edges.  
For the following epoch, say, the $i$th epoch, Observation~\ref{obs:initial-matching}
tells us that regardless of what matching we had at the beginning of the epoch,
the total number of reassignments performed by SAP during the epoch will not exceed $\Oo(n_i \log^2 n_i)\subseteq \Oo(n\log^2 n)$.
We thus make at most $\Oo(n \maxload \log^2n)$ reassignments in total,
which completes the proof if $L < \sqrt{n}/\log n$.
If $\maxload \geq \sqrt{n}/\log n$, we make $\Oo(n\sqrt{n}\log n)$ reassignments during the first 
$\sqrt{n}/\log n$ epochs.
In all future epochs, note that a server at its maximum allowable load has at 
least $\sqrt{n}/\log n$ clients
assigned to it, so there are at most $\sqrt{n}\log n$ such servers, and whenever a client is inserted the 
shortest augmenting path to a server below maximum load will have length $\Oo(\sqrt{n}\log n)$.
This completes the proof because there are only $n$ augmenting paths in total. 
\end{proof}

\subsection{Approximate semi-matching}
Though our result on minimizing maximum load \emph{exactly} is tight, we conclude this section on extensions with a short and cute improvement for \emph{approximate} load balancing, which follows from the Expansion Lemma (Lemma~\ref{lem:expand}).

We study a setting similar to that of \cite{BernsteinKPPS17}, in which one wishes to minimize not only the maximum load, but the $p$-norm $\left|X\right|_p = \left(\sum_{s\in S} l(s)^p\right)^{\frac{1}{p}}$, where $l(s)$ is the load of the server $s$ in the assignment $X$, for every $p\geq 1$.

First, observe that a lower bound on the $p$-norm comes from our necessity values $\alpha$ from Section~\ref{sec:server-flow}. That is, $\left(\sum_{s\in S} \alpha(s)^p\right)^{\frac{1}{p}}\leq \left|X\right|_p$, for any assignment $X$. 
For $p=1$, we even have equality, as we simply count the number of clients. For $p>1$, the proof is almost identical to that of uniqueness in Section~\ref{subsubsec:unique}.

In Section~\ref{sec:MinMax}, we saw that even for the $\infty$-norm, we cannot obtain $\lceil \alpha(s) \rceil$ with logarithmic recourse. This motivates the use of approximation, and motivates the following definition:

\begin{definition}[$(1+\varepsilon)$-approximate semi-matching] For each server s in the current graph, let $L(s) = \lceil (1+\varepsilon) \alpha(s) \rceil$. We say that a semi-matching is $(1+\varepsilon)$-approximate, if each server s is assigned at most $L(s)$ clients.
\end{definition}

Assigning $(1+\varepsilon)\alpha(s)$ clients to server $s$ would indeed give a $(1+\varepsilon)$-approximation for every $p$-norm. 
Unfortunately, $(1+\varepsilon)\alpha(s)$ may not be an integer, which is why we apply the natural ceiling operation. 
Under the further assumption that $\alpha(s)\geq \frac{1}{\varepsilon}$ for all $s$, we have
\begin{align*}
\frac{\ceil{(1+\varepsilon)\alpha(s)}}{\alpha(s)}<
\frac{(1+\varepsilon)\alpha(s)+1}{\alpha(s)} =
1+\varepsilon+\frac{1}{\alpha(s)} \leq 1+2\varepsilon
\end{align*}
And thus, under the assumption that all necessities are $\geq\frac{1}{\varepsilon}$, then for 
any $(1+\varepsilon)$-approximate assignment $X$, 
where we 
let $l(s)$ denote the number of clients assigned to server $s\in S$,
we have:
\begin{align*}
\left(\sum_{s\in S} l(s)^p\right)^{\frac{1}{p}} & \leq &  \left(\sum_{s\in S} \left\lceil(1+\varepsilon)\alpha(s)\right\rceil^p\right)^{\frac{1}{p}} & < & 
\left((1+2\varepsilon)^p\sum_{s\in S}\alpha(s)^p\right)^{\frac{1}{p}} & = & 
(1+2\varepsilon)\left(\sum_{s\in S}\alpha(s)^p\right)^{\frac{1}{p}}
\end{align*}
But as already noted, the $p$-norm of the $\alpha$-vector is a lower bound on any assignment, including the optimal assignment $X_{\opt}$, so $\left|X\right|_p \leq (1+2\varepsilon)\left|X_{\opt}\right|$.

In the following, let $n$ denote the number of clients that have arrived thus far.

\begin{theorem}$(1+\varepsilon)$-approximate semi-matching has worst-case $O(\frac{1}{\varepsilon} \log n)$ reassignments with SAP.
\end{theorem}

The proof of this theorem relies again on the Expansion Lemma.
In this case, however, we do not use the $\alpha$-values as part of an amortization argument, but only to bound the lengths of the shortest augmenting paths.
	
\begin{proof}
Given our graph $G$, let $G'$ denote a similar graph with $L(s)$ copies of each server.  Then any maximum matching in $G'$ corresponds to an $(1+\varepsilon)$-approximate semi-matching in $G$.
Now, note that each client-set $K$ in $G'$ has a neighborhood of at least $(1+\varepsilon)$ times its own size:
\begin{align*}
\hfill && \left| N_{G'} (K)\right| && = && \sum_{s\in N_G(K)} L(s) && \geq &&  (1+\varepsilon) \sum_{s\in N_{G}(K)}\alpha_G (s) && \geq && (1+\varepsilon)\left|K\right| && \hfill
\end{align*}	
Where the last inequality follows from the fact that the neighborhood of $K$ receives at least all the flow from $K$, and thus, at least $K$ flow.
Thus, by Lemma~\ref{lem:maxalpha} we can upper bound the highest alpha-value $\alphamax$ in $G'$ by
\begin{align*}
  \alphamax &= \max_{\emptyset\subset K\subseteq C}
  \frac{\sizeof{K}}{\sizeof{N_{G'}(K)}} \leq \frac{1}{1+\varepsilon} =
  1 - \frac{\varepsilon}{1+\varepsilon}
\end{align*}
By setting $\varepsilon^{\prime} = \frac{\varepsilon}{1+\varepsilon}$, all servers of $G'$ has necessity $\leq 1 - \varepsilon^{\prime}$. The Expansion Lemma (Lemma~\ref{lem:expand}) then gives that any active augmenting tail has length at most $\frac{2}{\varepsilon'}\ln (n) = 2\frac{1+\varepsilon}{\varepsilon}\ln n$, which is $O(\frac{1}{\varepsilon}\log n)$.
\end{proof}

\section{Conclusion}
We showed that in the online matching problem with replacements, where vertices on one side
of the bipartition are fixed (the servers), while those the other side arrive one at a time
with all their incoming edges (the $n$ clients), the shortest augmenting path protocol maintains a maximum matching
while only making amortized $\Oo(\log^2 n)$ changes to the matching per client insertion. 
This almost matches the $\Omega(\log n)$ lower bound of Grove et al.~\cite{Grove:95}.
Ours is the first paper to achieve polylogarithmic changes per client; the previous
best of Bosek et al. required $\Oo(\sqrt{n})$ changes, and used a non-SAP strategy~\cite{BosekLSZ14}.
The SAP protocol is especially interesting to analyze because it is the most natural greedy approach
to maintaining the matching. However, despite the conjecture of Chaudhuri et al.~\cite{conf/infocom/ChaudhuriDKL09} that the SAP protocol only requires $\Oo(\log n)$ amortized changes per client, our analysis is the first to go beyond the trivial $\Oo(n)$ bound for general bipartite graphs; previous results were only able to analyze SAP in restricted settings. Using our new analysis technique, we were also able to show an implementation of the SAP protocol that requires total update time 
$\Oo(m\sqrt{n}\sqrt{\log n})$, which almost matches the classic offline $\Oo(m\sqrt{n})$ running time
of Hopcroft and Karp~\cite{Hopcroft73}. 

The main open problem that remains is to close the gap between our $\Oo(\log^2 n)$ upper bound
and the $\Omega(\log n)$ lower bound. This would be interesting for any replacement strategy,
but it would also be interesting to know what the right bound is for the SAP protocol in particular.
Another open question is to remove the $\sqrt{\log n}$ factor in our implementation of the 
SAP protocol. Note that both of these open questions would be resolved if we managed
to improve the bound in Lemma~\ref{lem:bound-long-paths} from $\Oo(n\ln(n)/h)$ to $\Oo(n/h)$. 
(In the implementation of Section~\ref{sec:implementation} we would then set $h = \sqrt{n}$ 
instead of $h = \sqrt{n}\sqrt{\log n}$.)

\section{Acknowledgements}
The first author would like to thank Cliff Stein for introducing him to the problem. The authors 
would like to thank Seffi Naor for pointing out to us that uniqueness of server loads can be proved via convex optimization (Section~\ref{subsubsec:unique}), and to thank Martin Skutella and Guilamme Sagnol for very helpful pointers regarding the details of this proof.

\bibliographystyle{plain}
\bibliography{refs}

\end{document}